\documentclass[journal]{IEEEtran}
\usepackage{cite}
\usepackage{mciteplus}
\usepackage{citesort}
\usepackage{amsthm}
\theoremstyle{plain} 
\usepackage{tabularx}
\usepackage[flushleft]{threeparttable}
\usepackage{multirow}
\usepackage{booktabs}
\usepackage{xcolor}
\usepackage{graphicx}

\usepackage{epsfig}
\usepackage{epstopdf}
\usepackage{psfrag}
\usepackage{subfigure}

\usepackage{url}

\usepackage{stfloats}
\usepackage{amsmath}
\usepackage{amsfonts}
\usepackage{amssymb}
\usepackage{eqnarray}
\usepackage{subeqnarray}
\usepackage{flushend}

\usepackage{array}

\usepackage{setspace}
\usepackage{paralist}

\usepackage{algorithm}
\usepackage{algpseudocode}
\usepackage{graphicx}    
\usepackage{caption}     
\usepackage{subcaption}  

\usepackage{acronym}

\acrodef{sm}[SM]{spatial modulation}
\acrodef{ssk}[SSK]{space shift keying}
\acrodef{smx}[SMX]{spatial multiplexing}
\acrodef{ml}[ML]{maximum--likelihood}
\acrodef{sd}[SD]{Sphere Decoder}
\acrodef{aber}[ABER]{average bit error ratio}
\acrodef{smrx}[SM--Rx]{Receiver--centric SD}
\acrodef{smtx}[SM--Tx]{Transmit--centric SD}
\acrodef{mimo}[MIMO]{multiple--input multiple--output}
\acrodef{simo}[SIMO]{single--input multiple--output}
\acrodef{cdf}[CDF]{cumulative distribution function}
\acrodef{bpsk}[BPSK]{binary phase shift keying}
\acrodef{qpsk}[QPSK]{quadrature phase shift keying}
\acrodef{qam}[QAM]{quadrature amplitude modulation}
\acrodef{snr}[SNR]{signal-to-noise-ratio}
\acrodef{iid}[i.i.d.]{identical and independently distributed}
\acrodef{ici}[ICI]{inter--channel interference}
\acrodef{rf}[RF]{radio frequency}
\acrodef{nlos}[NLOS]{Non Line--of--Sight}
\acrodef{los}[LOS]{line of sight}
\acrodef{scr}[SC]{spatial correlation}
\acrodef{pep}[PEP]{pairwise error probability}
\acrodef{mgf}[MGF]{moment-generation function}
\acrodef{fbesm}[FBE--SM]{fractional bit encoding spatial modulation}
\acrodef{gsm}[GSM]{generalized spatial modulation}
\acrodef{csi}[CSI]{channel state information}
\acrodef{pdf}[PDF]{probability distribution function}
\acrodef{mrc}[MRC]{maximum ratio combining}
\acrodef{rv}[RV]{random variable}
\acrodef{awgn}[AWGN]{additive white Gaussian noise}
\acrodef{qsm}[QSM]{quadrature spatial modulation} 
\acrodef{mac}[MAC]{multiple access channel}
\acrodef{zmcgv}[ZMCGV]{zero--mean complex Gaussian vector}
\acrodef{dof}[DoF]{degree of freedom}
\acrodef{eda}[OSA]{optimally spaced antennas}
\acrodef{rda}[RSA]{randomly spaced antennas}
\acrodef{aoa}[AOA]{angle of arrival}
\acrodef{aod}[AOD]{angle of departure}
\acrodef{itu}[ITU]{International Telecommunication Union}
\acrodef{dsm}[DSM]{differential spatial modulation}
\acrodef{smt}[SMT]{space modulation techniques}
\acrodef{dco}[DCO]{Asymmetrically clipped DC biased optical}
\acrodef{aco}[ACO]{Asymmetrically clipped optical}
\acrodef{ofdm}[OFDM]{orthogonal frequency division multiplexing}
\acrodef{papr}[PAPR]{Peak average power ratio}
\acrodef{mamimo}[MaMIMO]{massive multiple input multiple output}
\acrodef{ofdm}[OFDM]{orthogonal frequency division}
\acrodef{das}[DAS]{distributed antenna systems}
\acrodef{thz}[THz]{Terahertz}
\acrodef{aps}[APs]{access points}
\acrodef{ues}[UEs]{user equipment's}
\acrodef{6}[6G]{sixth generation}
\acrodef{uav}[UAV]{unmanned aerial vehicle }

\newtheorem{theorem}{Theorem}
\newtheorem{corollary}{Corollary}

\begin{document}



\title{ Symbiotic FAS Strategies for 6G UAVs Assisted Backscatter Networks}


\author{\IEEEauthorblockN{Nagla Abuzgaia\IEEEauthorrefmark{5}\IEEEauthorrefmark{7}, Abdelhamid Salem\IEEEauthorrefmark{5}\IEEEauthorrefmark{4},\textit{ Member, IEEE}, Ahmed Elbarsha\IEEEauthorrefmark{5}, and Khaled Rabie\IEEEauthorrefmark{3},\textit{ Senior Member, IEEE}. }
\thanks{\IEEEauthorrefmark{5} Authors are with University of Benghazi, Electrical and Electronics Engineering Department, Faculty of Engineering, Benghazi, Libya, E-mails: \{{nagla.abuzgaia , abdelhamid.albaraesi \& ahmed.elbarsha\}@uob.edu.ly}}	
\thanks{\IEEEauthorrefmark{7}Nagla Abuzgaia is also with Libyan Authority for Scientific Research, Tripoli, Libya, E-mail: {nagla.abuzgaia@aonsrt.ly}\\ \IEEEauthorrefmark{4} Abdelhamid Salem is also with Department of Electronic and Electrical Engineering, University College London, UK, E-mail: {a.salem@ucl.ac.uk}}
\thanks{\IEEEauthorrefmark{3} Khaled Rabie is with the Department of Computer Engineering, King Fahd University of Petroleum and Minerals (KFUPM), Dhahran, Saudi Arabia, Email: k.rabie@kfupm.edu.sa.}
}

\maketitle

\begin{abstract}
This paper investigates a fluid antenna system (FAS) enabled symbiotic radio (SR) network featuring a UAV mounted FAS communicating through a wireless power transfer (WPT) signal with a remote cluster head and an ambient tag. To evaluate system reliability, we derive the upper (UB) and lower (LB) bounds for both the composite and backscatter (BcS) outage probabilities, formulate the coexistence outage probability (COP), and present an asymptotic analysis that explicitly characterizes the system's spatial diversity and coding gains. We propose three novel symbiotic strategies; maximum backscatter selection (MBS), joint balanced selection (JBS) and threshold aware priority selection (TAPS) and compare them with the conventional maximum composite gain selection (MCGS) and random selection methods. Under a joint optimization framework, the macroscopic UAV 2D spatial placement and microscopic realization-level port selection were formulated.  Since the joint outage Pareto frontier is highly non-convex, the computationally expensive $\epsilon$-constraint method identified the optimal knee-point. While the symbiotic novel TAPS strategy demonstrated identical performance in a single step with linear $\mathcal{O}(MN)$ complexity. Moreover, an asymptotic closed-form mobility analysis under Jakes' fading model proves that the symbiotic FAS tracking protocol consumes under $10\%$ of the channel coherence time. Simulation results and COP heat maps validate the optimal coexistence performance at the ideal UAV coordinates with $60\%$ reduction in the required transmit power.

\end{abstract}

\begin{IEEEkeywords}
Fluid antenna system (FAS), UAVs, symbiotic radio (SR), backscatter (BcS), COP, Pareto front.  
\end{IEEEkeywords}

\IEEEpeerreviewmaketitle

\section{Introduction}

\IEEEPARstart{U}{ncrewed} aerial vehicles (UAVs) offer on-demand deployment and flexible three-dimensional connectivity, making them attractive enablers for sparse, remote, and infrastructure-poor networks \cite{Granelli2026}. This capability is particularly transformative for rural hard to reach area that needs wireless coverage or/and wireless powered communication networks (WPCNs) services, where  UAVs can serve as an aerial base station/mobile power beacon, wirelessly charging remote internet of things (IoTs) devices that subsequently transmit their data back to the UAVs or to a dedicated receive hub \cite{Xie2021}. This poses UAVs communication as a vital corner stone in sixth generation (6G) networks toward ubiquitous connectivity. 
However, UAVs suffers from limited onboard energy resources and capabilities to install heavy computations hardwares or antennas, known as size weight and power (SWAPs) constraints. This hinder UAVs from mounting large heavy  multiple input multiple output (MIMO) fixed position antennas (FPAs) structures, limiting the degree of freedom (DoF) provided through the diversity and multiplexing gains of MIMOs different schemes.   \\

To counter FPAs limitations, novel fluid antenna system (FAS) technology was proposed in \cite{Wong2020}. FAS emerged as a technology that leverage the flexibility of port selection within a compact spatial distance to harness the diversity gain \cite{Psomas2023}, beam forming gain \cite{Xu2026beamforming}, and more recently the pixel antenna coding gain \cite{shen2026coding}. Because of these gains and FAS convenient characteristics under UAVs SWAPs constrains, authors has been investigation the viability of implementing FAS on-board of UAVs\cite{Abuzgaia2026FAS}\cite{Abuzgaia2026localization}. In \cite{Abuzgaia2026FAS}, FAS proved to provide energy efficient margin for UAVs when investigated in WPCNs. While in \cite{Abuzgaia2026localization}, authors demonstrated the capabilities of FAS geometric gain to enhance the UAVs self localization capabilities in denies global positioning system (GPS) urban corridors. Another important aspect in FAS research is developing practical protocols with convenient FAS ports density that can handle the construction of FAS channel state information(CSI)\cite{Elganimi2026}, and port selection within permissible latency\cite{Zhu2026} and within the channel coherence symbol block\cite{Mu2026LoRaFAS}. In \cite{Elganimi2026} authors investigated the slow FAMA channel construction when FAS is implemented at the user side. In \cite{Zhu2026} authors investigated the user side FAS port switching delay effect on the performance of hyper-reliable low-latency communications (HRLLC) under finite blocklength operation. In\cite{Mu2026LoRaFAS} a long-range communication (LoRa)-FAS protocol for IoT networks is proposed which embed pilot sequences directly into symbols instead of using separate preambles, reducing overhead and physical-layer frame latency with the condition of a long channel coherence time. However, these protocols discuss the scenarios of receiver mounted FAS. Moreover, current research lacks studies on FAS mounted UAVs specific mobility spatial effects on protocols feasibility\cite{Abuzgaia2026FAS}. \\

In parallel,  symbiotic radio (SR) has emerged as a promising energy efficient paradigm for mutualism spectrum sharing, in which passive backscatter (BcS) devices reuse the primary radio frequency (RF) waveform to communicate with little or no additional spectrum or power\cite{Liang2020}. Therefore, merging SR paradigm into UAVs communications structure has been investigated recently toward better coverage \cite{Yang2025} and spectrum use\cite{Xu2026IntSenSym}, especially for remote or hard-to-reach areas. A recent paper explicitly frame the goal as improving network coverage while increasing spectrum utilization \cite{Yang2025}. Another paper on UAVs swarm utilizes symbiotic communications technique for spectrum, energy efficient communications and sensing within the swarm, where a leading UAV acts as a primary transmitter to the ground base station while the other UAVs in the swarm use this signal to backscatter their data \cite{Xu2026IntSenSym}. However, UAV-assisted symbiotic links are still limited by strong direct-link interference, difficult channel estimation, mobility-induced non stationariness, and the need to jointly satisfy backscatter throughput, harvested-energy, and QoS constraints\cite{Zhong2024,Janjua2024,Zhang2025}.\\

Nevertheless, SR performance optimization problems are characterized with multi-objective non-convex formulations that needs iterative solvers with high computations complexity. The SR rate performance optimization for the primary and secondary links were performed through the beam forming gain optimization of the movable antenna \cite{Zhou2024}. While the SR energy efficiency Pareto-boundry was characterized using the weighted linear sum algorithm and the optimization of the transmit beamforming and power allocation were performed by the SCA algorithm \cite{Wang2022}. More recently, a Pareto-boundry for the capacity region of FAS assisted SR (FAS-SR) was formulated as a standard linear weighted sum problem and optimized with a chaotic sequence-based adaptive particle swarm optimization (CSA-PSO)\cite{Li2026FASR}.  However, the capacity region Pareto-boundry of  FAS-SR under correlated and discrete ports is a non-convex boundary and linear scalarization completely fails to find solutions in non-convex regions of a Pareto front as the linear sweep will bridge across concave entirely missing the actual optimal ports. \\ 

Motivated by these challenges, this paper investigates a novel FAS-SR network architecture. The proposed system model features an UAV equipped with an $N$-port FAS that dynamically communicates with an ambient IoT tag. At the receiving terminal, a remote cluster head (CH) handles the processing workload by executing both the composite signal energy harvesting (EH) assessments and the BcS data stream decoding. To guarantee network reliability, we construct a joint multi-objective optimization framework that bridges two distinct operations: the 2D spatial placement of the UAV platform and the symbiotic FAS port selection policy. This unified structural framework has successfully overcome the non-convex physical trade offs of the symbiotic channel. To this end, the core contributions of this paper are:\\

1) Investigation for the first time FAS heuristic port selection capabilities in fulfilling the symbiotic trade-off optimum performance under WPCNs scenario. As novel symbiotic FAS selection strategies are proposed to tackle the reliability performance in the downlink (DL) of both the composite received power for the energy harvesting (EH) at the CH and the BcS link for IoTs data detection. These symbiotic FAS strategies achieve a near-perfect multi-objective compromise, minimizing BcS data outage while strictly defending the EH boundary without requiring highly complex, real-time iterative weight-tuning of the multi-objective solvers. In fact, these symbiotic FAS strategies are completely transferable to other SR scenarios. \\   

2) In the protocol complexity and mobility feasibility analysis subsection, we provide a rigorous, closed-form mathematical verification proving the real-time physical viability of executing multi-port FAS tracking from a moving UAV platform when data decoding occurs at a remote cluster head (CH). By deriving a generalized time-budget boundary condition under Jakes' isotropic scattering model, we analytically prove that the total signaling, computation, and feedback overhead consumes less than $10\%$ of the channel's coherence envelope. Furthermore, our asymptotic scaling analysis establishes a clear engineering design law: the spatial resolution gains achieved by expanding FAS ports $N$ introduce a manageable $\mathcal{O}(1/N)$ trade-off with the maximum permissible flight velocity of the aerial node.\\

2) Novel symbiotic multi objective FAS port selection strategies is proposed and their global outage probability lower (LB) and upper bounds (UB) performance characterized under Nakagami-m fading channels with coherent/SIC detector for BcS data. Moreover, asymptotic bounds outage probability are characterized for the composite and BcS links, where the diversity gain order and coding gain of FAS are identified theoretically. Furthermore, coexistence outage probability (COP) bounds were derived and FAS mounted UAV placement optimization for optimum COP performance. \\

4) We characterize the fundamental multi-objective trade-offs governing the symbiotic network by mapping the joint reliability region of the primary composite link and the dyadic backscatter link. We demonstrate that conventional linear weighted-sum scalarization fails to capture the structurally non-convex boundaries of FAS performance. To resolve this, we leverage a weighted Tchebycheff scalarization framework alongside an $\epsilon$-constraint baseline to establish the exact Pareto-front boundary. To enable real-time execution, we propose threshold aware priority selection (TAPS) strategy, which acts as a single step realization level filter. TAPS accelerates directly to the optimal knee-point of the Pareto frontier, achieving the performance of an optimized $\epsilon$-constraint framework without requiring computationally expensive iterative loops.\\

The rest of the paper is organized as follows; Section.II presents the system model with FAS mounted UAV channel and symbiotic signals ratio, Section.III proposes novel symbiotic FAS strategies with complexity and feasibility analysis, Section.IV provides the outage bounds derivations,  Section.V presents multi-objective framework through Pareto front boundary formulation and UAV placement optimization, and Section.VI presents the simulation results, to end up with the conclusion in Section.VII.

\section{System model}
We consider an UAV hovering in a time slot $T$, representing an aerial base station at the position coordinate $\mathbf{p_{\text{UAV}}} = (\text{X}_{\text{u}}, \text{Y}_{\text{u}},\text{h}_{\text{u}})^T \in \mathbb{R}^{3 \times 1}$ equipped with $N$-ports fluid antenna that transmits energy during $\alpha T$ to a designated ground node, which acts as a data aggregator for a cluster of sensors \cite{Wei2022}, called cluster head (CH) and positioned at $\mathbf{p_{CH}}=(\text{X}_{\text{CH}},\text{Y}_{\text{CH}},\text{ h}_{\text{CH}})^T \in \mathbb{R}^{3 \times 1}$, the distance between UAV and CH is the direct link $\text{d}_{\text{D}}=\|\mathbf{p}_{\text{CH}} - \mathbf{p}_{\text{UAV}}\|$. Similarly, the A-IoT tag is located at $\mathbf{p_{\text{Tag}}}=(\text{X}_{\text{t}},\text{Y}_{\text{t}},\text{h}_{\text{t}})^T \in \mathbb{R}^{3 \times 1}$, the distance between the UAV and the A-IoT tag is the BcS forward link $\text{d}_{\text{f}}=\|\mathbf{p}_{\text{Tag}} - \mathbf{p}_{\text{UAV}}\|$, while the BcS backward distance to the CH is $\text{d}_{\text{g}}=\|\mathbf{p}_{\text{CH}} - \mathbf{p}_{\text{Tag}}\|$. The CH node employs a hybrid energy harvesting model, where the solar energy provides the base power for the SIC processing engine, while the UAV-mounted FAS optimizes the RF link to maximize backscatter signal strength and provide supplemental WPT, ensuring energy-neutral operation even in fluctuating environmental conditions. The collected tag's BcS data are transmitted back to the UAV during the remaining of the time slot i.e. $(1-\alpha)T$, as was studied by authors in recent study \cite{Abuzgaia2026FAS}, and demonstrated here in Fig.\ref{fig:FA_UAV_WPT}. Nevertheless, we devote this work solely for DL performance as it represents the bottle neck of this symbiotic system. 

\begin{figure}[h!]
 \includegraphics[width=7cm]{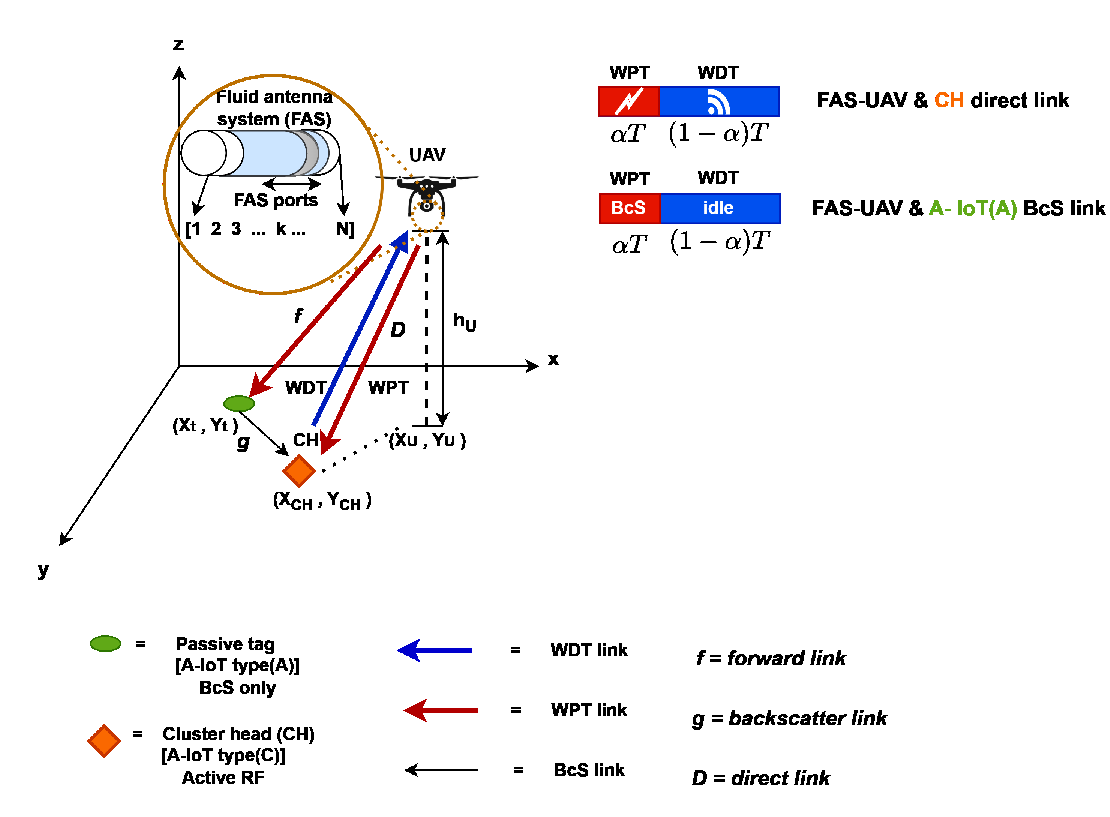}\vspace{-0.2cm}
 \caption{System model of a FAS mounted UAV assisted symbiotic WPCNs.}
 \label{fig:FA_UAV_WPT}
\end{figure}

\subsection{FAS-UAV channel model}
For analytical tractability, we assume the channel model as time division (TD). Due to the small size of the spacing between the N ports of the fluid antenna compared with the distance between the UAV and the CH we could approximate the $ k^{th}$ port pathloss of  the direct DL and the backscattered forward link as the same for all the ports, i.e. $L_k(d) \approx L(d) $, hence the DL-WPT channel model and the BcS forward link channel model, respectively, are \cite{Liu2021}
\begin{equation}
\sqrt{L_i(d_i)} h_{i,\bar{k}} =\sqrt{P_u (d_i^{-\rho})}h_{i,\bar{k}}, \, \, i \in \{D,f\}, 
\end{equation}
where $P_u$ is the channel power gain at the reference distance, $\rho$ is the path loss exponent, while $D$ and $f$ are the FAS-UAV to CH's WPT direct link and FAS-UAV to  A-IoT's BcS forward link, respectively, with $\bar{k}$ as the activated port for both links. The FAS ports fading channels are assumed to follow Nakagami-m distribution with spatial correlation as follows \cite{Mavrovoltsos2024}:
\begin{equation}
\begin{split}
H_{kl} = \sqrt{1 - \mu^2_k} x_{kl} + \mu_k x_{0l} + j \left( \sqrt{1 - \mu^2_k} y_{kl} + \mu_k y_{0l} \right),\\
  l = \{1, \dots, m\},
 \end{split}
\end{equation}
where  $\{x_{kl}$, $ y_{kl}\}$ are independent Gaussian r.v.s with zero mean and 1/2 variance, $\{x_{ol}$, $ y_{ol}\}$ is the reference port, and $\mu_k$ is the correlation coefficient which follows \cite{Wong2022}: 
\begin{equation}
\mu^2 = \frac{2}{N(N - 1)} \sum_{k=1}^{N-1} (N - k) J_0 \left( \frac{2\pi k W}{N - 1} \right), \quad \text{for } \mu_k = \mu \forall k,
\label{eq:mu_squared}
\end{equation}
where $J_0( )$ is the zero-order Bessel function of the first kind. Then the direct WPT and forward BcS links correlated $k^{th}$ port Nakagami-m fading channels envelops with normalized variance are \cite{Mavrovoltsos2024}, 
\begin{equation}
|h_i,k| = \sqrt{\sum_{l=1}^{m} \frac{1}{m} |H_{i,kl}|^2}, \, \, i \in \{D,f\}.
\end{equation} 
On the other hand, $|h_g|$ is the A-IoT to CH BcS link magnitude with distance $d_g$, and defined as a Nakagami-$m$ single variable that's common to all ports. 

\subsection{Symbiotic signals \& SNR}
The proposed symbiotic system in this paper consists of two main links; the direct link which represents the FAS-UAV's WPT unmodulated carrier signal used for the energy charging of the CH, along side the backscatter (BcS) link in which the A-IoT tag uses the direct link signal to backscatter their data to the CH \footnote{The system model could be extended into a direct link signal that carries both the WPT signal and communication data from the UAV, i.e. SWIPT.}.
To first estimate the direct link signal at the CH, the A-IoT tags are kept inactive during the pilot symbol. Hence, the received direct link signal is: 
\begin{equation}
y^{CH} _{D,k}= \sqrt{P_{tx}}\cdot L_D\, h_{D,k} \, s+ n,
\end{equation}
where $P_{tx}$ is the transmitted power of the UAV, $s$ is the normalized carrier symbol with zero mean and unit variance, and ${n}$ is the additive white Gaussian noise (AWGN) with zero mean and $(N_o)$ variance.\\

By activating the A-IoT tag, the WPT carrier signal is captured by the A-IoT tag and modulated with BcS data $c$. Thus, the composite received signal at the CH is \cite{Ghadi2025}: 
\begin{equation}
y^{CH} _{C,k}= \sqrt{P_{tx}}\cdot L_D\, h_{D,k} \, s + \vartheta \sqrt{P_u}L_f L_g \, h_{f,k}h_g\, c \, s + n,
\label{eqn:y_c}
\end{equation}
where $\vartheta$ is the BcS reflection efficiency. From (\ref{eqn:y_c}), the equivalent composite signal channel is: 
\begin{equation}
|h_{C,k}| =  |L_D \, h_{D,k}+ \vartheta L_f L_g \, h_{f,k}h_g\, c|,
\end{equation}
where the BcS link signal channel is: 
\begin{equation}
|h_{BcS,k}| = |\vartheta L_f L_g \, h_{f,k}h_g\, c|.
\end{equation} 

Therefore, the normalized composite signal received power is\footnote{For analytical convenience and to provide a unified scale with the data detection threshold, we normalize the received power by the thermal noise floor $N_0$.} 
\begin{equation}
\gamma_C = \bar{\gamma}_C |L_D \, h_{D,k}+ \vartheta L_f L_g \, h_{f,k}h_g\, c|^2,
\label{eqn:avg_c}
\end{equation}
where $\bar{\gamma}_C = \frac{P_{tx}}{N_o}$ is the transmitted composite power normalized to the thermal noise, while the normalized WPT direct link signal power is: 
\begin{equation}
\gamma_D=\bar{\gamma}_D |h_{D,k}|^2,
\label{eqn:avg_d}
\end{equation}
where $\bar{\gamma}_D=\frac{ P_{tx} \cdot L^2_D }{N_o}$ is the normalized average power of the WPT direct link signal.
On the other hand, the BcS link SNR: 
\begin{equation}
\gamma_{BcS}= \bar{\gamma}_{BcS}|h_{f,k}h_g|^2,
\label{eqn:avg_BcS}
\end{equation}
where $\bar{\gamma}_{BcS}=\frac{P_{tx} \cdot \vartheta^2 L^2_f L^2_g\,  \, c^2}{N_o}$ is the BcS link average SNR. 

\section{FAS port selection strategies}
\label{sec:FAS_symb} 
To harness the spatial diversity gain of FAS and utilize the novel DoF provided by the N-ports, the selection process is based on finding the port that provide the maximum channel gain within the available N ports. If the target was directly maximizing the received signal, either a single channel link as was demonstrated in our previous work \cite{Abuzgaia2026FAS}, or a composite signal as in \cite{Ghadi2025}\cite{Li2026FASR} then the selection strategy is straight forward as clarified next. 

\subsection{Conventional FAS strategies}
 
\subsubsection{Maximum composite gain selection (MCGS)}
In this strategy the selected port will be based on the port with the maximum composite signal energy at the CH, as follows: 
 \begin{equation}
\bar{k}_{MCGS}= \arg \max_{k\in \mathcal{N}} (|h_{C,k^*}|^2),
\end{equation}
which has been investigated before in recent literatures where the BcS link was an exterior support link for the main wireless data direct link between the base station and the user device, i.e. no BcS data was transferred in these literatures\cite{Ghadi2025}\cite{Li2026FASR}.\\
\begin{equation}
\bar{k}_{MCGS}=  \arg \max_{k\in \mathcal{N}}(|\sqrt{\bar{\gamma}_D} h_D +\sqrt{\bar{\gamma}_{BcS}} h_{BcS}|^2). 
\label{eqn:MCGS}  
\end{equation}
By expanding (\ref{eqn:MCGS})\cite{Radic2025}\cite{Shen2016},
 
\begin{equation}
\begin{split}
\bar{k}_{MCGS}=  \arg \max_{k\in \mathcal{N}} \Big(\bar{\gamma}_D|h_{D,k}|^2 + \bar{\gamma}_{BcS} |h_{f,k} h_g|^2 \\ + 2\sqrt{\bar{\gamma}_D \bar{\gamma}_{BcS} }|h_{D,k}| |h_{f,k} h_g| \cos(\Delta \phi)\Big),
\end{split}
\label{eqn:MCGS_phi}
\end{equation}
where $\Delta \phi $ is the phase offset between the direct and BcS signal, which is a uniform random phase $\Delta\phi \sim \mathcal{U}[0, 2\pi]$. \\

 Moreover, maximizing the summation composite signal at the CH generate the maximum possible diversity gain of FAS  \cite{Ghadi2025}, but doesn't guarantee the backscatter individual link channel gain performance which might be in deep fade while the maximum gain is sole due to the direct link. Hence, novel port selection methods are proposed for this symbiotic backscatter link to ensure that the backscatter link is detectable at the CH and not under deep feed while the direct link high gain dominate the maximized gain selection.\\
 
\subsubsection{Random Selection (RS)} 
This port selection strategy will be used as a benchmark, where the activated port will be selected randomly without any criteria to present the performance of the fixed antenna systems\cite{Mavrovoltsos2024}.
\subsection{Novel Symbiotic FAS strategies}
 Contrary to the single link channel, the symbiotic system in this paper contains a composite signal from two links; the direct WPT carrier signal and the backscattered tag signal. Hence, we propose in this section three different port selection strategies based on the system symbiotic priority. \subsubsection{Maximum backscatter selection (MBS)}
This strategy focusses on optimizing the BcS link within the composite signal using the received composite and direct signals gain as follows:   
\begin{equation}
\bar{k}_{MBS}= \arg \max_{k\in \mathcal{N}} ( \bar{\gamma}_C \,|h_{C,k^*}|^2-\bar{\gamma}_D \, |h_{D,k^*}|^2).
 \label{eqn:MBS} 
\end{equation}
For performance evaluation purposes, a maximum backscatter selection lower bound ($\text{MBS}_{\text{LB}}$) strategy is proposed as follows: 
 \begin{equation}
\bar{k}_{MBS_{LB}}= \arg \max_{k\in \mathcal{N}} (|h_{BcS,k^*}|^2),
\label{eqn:MBS_LB} 
\end{equation}
substituting with (\ref{eqn:avg_BcS}) yields 
\begin{equation}
\bar{k}_{MBS_{LB}}= \bar{\gamma}_{BcS}|h_g|^2  \arg \max_{k\in \mathcal{N}}(|h_{f,k^*}|^2).
\label{eqn:MBS_LB_forwrd}
\end{equation}

\subsubsection{Joint balanced selection (JBS)}
Representing a threshold blind joint optimization scheme, JBS balances the direct and backscatter links simultaneously without requiring explicit knowledge of the decoding thresholds. It achieves this by activating the FAS port that maximizes the product of the individual link powers as follows:
\begin{equation}
\bar{k}_{\text{JBS}} = \arg \max_{k \in \mathcal{N}} \left( \bar{\gamma}_D |h_{D,k}|^2 \cdot \bar{\gamma}_{\text{BcS}} |h_{\text{BcS},k}|^2 \right).
\label{eqn:JBS}
\end{equation}
While JBS avoids a singular symbiotic task optimization as in MCGS and MBS, its threshold blind nature can lead to activating a port that exceeds the BcS reliability threshold over accessible composite signal performance enhancement, leading to non optimal symbiotic solution.\\

\subsubsection{Threshold aware priority selection (TAPS)}

Let $\mathcal{K}_{\text{valid}}$ denote the dynamic subset of FAS ports that successfully satisfy the minimum BcS reliability constraint during a given channel instantiation, defined as:
\begin{equation}
\mathcal{K}_{\text{valid}} = \left\{ k \in \mathcal{N}, \;\big|\; \bar{\gamma}_{\text{BcS}} |h_{\text{BcS},k}|^2 \ge \gamma_{\text{min}} \right\}.
\end{equation}

The TAPS port selection policy operates on a conditional priority rule. If at least one port satisfies the backscatter threshold requirement ($\mathcal{K}_{\text{valid}} \neq \emptyset$), the strategy allocates the remaining spatial degree-of-freedom to maximize the primary composite link. Conversely, if the fading environment is highly severe and no single port satisfies the target threshold ($\mathcal{K}_{\text{valid}} = \emptyset$), the algorithm switches to a fallback security mode that maximizes the BcS gain to minimize the outage deficit. The selected port $\bar{k}_{\text{TAPS}}$ is expressed as:
\begin{equation}
\bar{k}_{\text{TAPS}} = \begin{cases} 
\arg \max_{k \in \mathcal{K}_{\text{valid}}} \left( \bar{\gamma}_C |h_{C,k}|^2 \right), & \text{if } \mathcal{K}_{\text{valid}} \neq \emptyset, \\ 
\arg \max_{k \in \mathcal{N}} \left( \bar{\gamma}_{\text{BcS}} |h_{\text{BcS},k}|^2 \right), & \text{if } \mathcal{K}_{\text{valid}} = \emptyset.
\end{cases}
\end{equation}

For real time embedded implementation on a UAV platform, this conditional logic can be mapped into a vectorized, penalty based objective function requiring zero iterative loops, formulated in (\ref{eqn:TAPS}) below, where $L \gg 1$ represents a large positive regularizing scaling bias, and $\mathbb{I}_k \in \{0,1\}$ is the indicator variable checking the threshold state of port $k$, defined by:
\begin{equation}
\mathbb{I}_k = \frac{1}{2} \left( 1 + \text{sgn}\left( \bar{\gamma}_{\text{BcS}} |h_{\text{BcS},k}|^2 - \gamma_{\text{min}} \right) \right).
\end{equation} 
 
\begin{figure*}[h]
\begin{equation}
\bar{k}_{\text{TAPS}} = \arg \max_{k \in \mathcal{N}} \left[ \left( \bar{\gamma}_C |h_{C,k}|^2 + L \right) \cdot \mathbb{I}_k + \left( \bar{\gamma}_{\text{BcS}} |h_{\text{BcS},k}|^2 \right) \cdot (1 - \mathbb{I}_k) \right],
\label{eqn:TAPS}
\end{equation}
\hrulefill 
\end{figure*}

\subsection{Implementation complexity and feasibility analysis}
Since the proposed symbiotic FAS policies require sequential spatial port evaluation across FAS dimension $N$, establishing its physical viability within the UAV time varying wireless channels is critical. In this subsection, we derive the generalized boundary conditions under which the signalling overhead of the symbiotic tracking framework remains strictly upper bounded by the channel's coherence envelope.
\subsubsection{Symbiotic protocol phases}
Let $T_s$ denote the fundamental symbol duration of the system. The complete port selection and feedback sequence is executed over a structured frame consisting of three distinct phases:
\begin{enumerate}
    \item \textbf{Spatial piloting phase:} The UAV sequentially traverses across all $N$ ports. At each port $k \in \mathcal{N}$, a pilot frame consisting of $N_D$ symbols is transmitted to estimate the direct channel state at the CH, followed by $N_B$ symbols dedicated to backscatter link illumination from the tag and composite channel estimation at the CH. Accounting for the transient switching duration $T_{\text{sw}}$ of the RF micro-electromechanical systems (MEMS) or solid-state switch matrix per port transition, the cumulative estimation delay $\tau_{\text{est}}$ is analytically formulated as:
    \begin{equation}
    \tau_{\text{est}} = N \left[ (N_D + N_B) T_s + T_{\text{sw}} \right].
    \end{equation}
    
    \item \textbf{CH computing phase:} Upon capturing the realization-level gains, the CH processes the intended symbiotic FAS objective function of (MCGS, MBS, JBS, or TAPS) defined in (\ref{eqn:MCGS_phi},\ref{eqn:MBS},\ref{eqn:JBS},\ref{eqn:TAPS}), respectively. Because these policies rely on a non-iterative conditional logical mask rather than an exhaustive heuristic search, its computational complexity scales strictly as $\mathcal{O}(N)$ comparison operations. Let $\tau_{\text{comp}}$ represent the hardware execution latency of the CH processor.
    
    \item \textbf{Feedback transmission phase:} To update the UAV port state, the CH quantizes the optimal port index $\bar{k}_{\text{symbiotic}}$ into a control packet of $\lceil \log_2(N) \rceil$ bits. Given a dedicated feedback control channel operating at a transmission rate of $R_{\text{fb}}$, and denoting the instantaneous propagation distance between the UAV and the CH as $\text{d}_{\text{D}}$, the total feedback and propagation latency $\tau_{\text{fb}}$ is given by:
    \begin{equation}
    \tau_{\text{fb}} = \frac{\lceil \log_2(N) \rceil}{R_{\text{fb}}} + \frac{\text{d}_{\text{D}}}{c},
    \end{equation}
    where $c$ is the speed of light in a vacuum.
\end{enumerate}

By summing these components, the total time overhead $\tau_{\text{total}}$ consumed by a single complete symbiotic FAS execution loop is bounded by:
\begin{equation}
\tau_{\text{total}} = N(N_D + N_B)T_s + NT_{\text{sw}} + \tau_{\text{comp}} + \frac{\lceil \log_2(N) \rceil}{R_{\text{fb}}} + \frac{\text{d}_{\text{D}}}{c}.
\end{equation}

\subsubsection{Generalized coherence boundary condition}
To preserve the validity of the selected port optimization, the entire protocol latency $\tau_{\text{total}}$ must be a small fraction of the channel coherence time $\tau_c$, during which the small-scale fading coefficients remain approximately invariant. Under Jakes' isotropic scattering model and for a low speed UAV, the coherence time is governed by the maximum Doppler shift $f_d$ \cite{Siying2026}, such that:
\begin{equation}
\tau_c \approx \frac{0.423}{f_d} = \frac{0.423 \, c}{v f_c},
\end{equation}
where $f_c$ is the carrier frequency, $v$ represents the instantaneous velocity of the UAV. We define the normalized protocol overhead fraction as $\eta = {\tau_{\text{total}}}/{\tau_c}$. For the symbiotic FAS system to be considered strictly practical, the system parameters must satisfy the following fundamental inequality:
\begin{equation}
\eta = \frac{v f_c}{0.423 \, c} [\tau_{\text{total}}] \le \eta_{\text{max}},
\label{eqn:overhead fraction}
\end{equation}
where $\eta_{\text{max}} \in (0, 1)$ represents the maximum permissible fraction of the coherence block allocated to control signaling (e.g., $\eta_{\text{max}} \le 0.1$).\\

\subsubsection{Asymptotic scaling and mobility insights}
Rearranging (\ref{eqn:overhead fraction}) yields the generalized maximum permissible operational velocity $v_{\text{max}}$ of the FAS-mounted UAV as a function of port density and network latency constraints:
\begin{equation}
v_{\text{max}} = \frac{0.423 \, c \cdot \eta_{\text{max}}}{f_c [\tau_{\text{total}}]}.
\label{eqn:vmax}
\end{equation}

From (\ref{eqn:vmax}), we establish the following structural properties of the proposed network architecture:
\begin{itemize}
    \item \textbf{Port Scalability Limits:} The maximum velocity exhibits an asymptotic scaling behavior of $v_{\text{max}} = \mathcal{O}\left(\frac{1}{N}\right)$ with respect to the FAS ports. This analytically demonstrates that the spatial resolution gains achieved by expanding $N$ introduce a linear trade-off with the maximum allowable mobility of the aerial node.
    \item \textbf{Decoupled Processing Advantage:} Because $\tau_{\text{comp}}$ is constrained by $\mathcal{O}(N)$\footnote{This complexity estimation is accelerated, as the required sampling points for accurate CSI estimation at the received FAS are proven way less due to ports high correlation effect \cite{Elganimi2026}, similarly not all of FAS ports need activation for accurate port index estimation.} rather than the exponential or combinatorial growth curves found in non-convex multi-objective solvers\cite{Li2026FASR}, shifting the processing workload to the CH allows the system to scale without experiencing a computational bottleneck. This optimization keeps the linear expansion profile of the total time overhead intact.
\end{itemize}

\section{Outage bounds performance analysis}
The symbiotic FAS system model in this paper addresses the symbiotic bottleneck of fulfilling the BcS data collection at the CH while also receiving the direct link signal power for energy harvesting as in the WPCNs. Thus, for generalized results we consider the symbiotic signals DL outage bounds performance. Moreover, the CH is modelled with two uncoupled distance antennas for each of the symbiotic circuits as the resource allocation effect of the time/power splitting circuits in WPCNs is out of the scope of this paper. Furthermore, a coherent receiver for the CH with ideal SIC of the direct signal is utilized to insure the detaching of direct signal effect from the BcS signal, guaranteeing a clear demonstration of the upper/lower FAS gain bounds effect.  
\subsection{Composite signal outage probability}
The FAS symbiotic strategies effects the coexistence reliability of both the composite signal received power for CH energy harvesting and the BcS signal SNR for tags data detection. Therefore, the normalized composite signal received power outage probability is derived here to characterize the reliability effect on the CH energy harvesting task. Moreover, the constructive/destructive addition of the direct link signal with the BcS link signal at the CH effects the performance of the symbiotic systems. Thus, substituting in (\ref{eqn:avg_c}) the normalized instantaneous received power \cite{ElMossallamy2018} at the CH after expanding is \cite{Radic2025}\cite{Shen2016}: 
 \begin{equation}
\begin{split}
\gamma_{CH,k}^{inst} = \bar{\gamma}_D|h_{D,k}|^2 + \bar{\gamma}_{BcS} |h_{f,k} h_g|^2+ \cdots \\ + 2\sqrt{\bar{\gamma}_D \bar{\gamma}_{BcS} }|h_{D,k}| |h_{f,k} h_g| \cos(\Delta \phi),
\end{split}
\label{eqn:inst_gamma}
\end{equation}
where $\bar{\gamma}_D$ and $\bar{\gamma}_{BcS}$ are the average normalized power of the direct link and BcS link SNR defined in (\ref{eqn:avg_d}) and (\ref{eqn:avg_BcS}), respectively, and $\Delta \phi $ is the phase offset between the direct and BcS signal, defined as a uniform random phase $\Delta\phi \sim \mathcal{U}[0, 2\pi]$. The maximization of this normalized composite received signal power depends on three different variables with non-linear relations, namely; $|h_{D,k}|$, $|h_{f,k}|$ and $\cos(\Delta\phi)$. Therefore, the derivation of the CDF of this symbiotic FAS correlated channels is tedious. In response, the  the global UB and LB equations of the composite signal normalized received power outage is derived next.

\begin{theorem}
 Global UB of the composite signal normalized received power outage probability for any FAS selection strategy is presented in the integral form as:
\begin{equation}
P_{out,C}^{UB} = \int_0^\infty \int_0^\infty \left(1 - \frac{\text{Re}\{\arccos(\beta)\}}{\pi} \right) f_U(u) f_V(v) \, du \, dv,
\label{eqn:Pout_UB}
\end{equation}
where $f_U(u)=f_{|h_D|}(x)$ is the direct link pdf of Nakagami-$m_D$ distribution with normalized power $\Omega_D=1$, and $ f_V(v)=f_{|h_{BcS}|}(v)$ is the backscatter total link product pdf of double Nakagami distribution with normalized power $\Omega_{BcS}=1$. 
\begin{subequations}
\begin{align}
f_{|h_D|}(u) = \frac{2 m_D^{m_D}}{\Gamma(m_D) \Omega_D^{m_D}} u^{2m_D-1} \exp\left( -\frac{m_D}{\Omega_D} u^2 \right),
\label{eqn:f_hD}\\
f_{|h_{BcS}|}(v) = \frac{4 \,  \mu^{\frac{m_f+m_g}{2}} \, v^{m_f+m_g-1} }{\Gamma(m_f)\Gamma(m_g)}   K_{m_f-m_g}\left(2v\sqrt{\mu}\right),
\label{eqn:f_hBcS}
\end{align}
\end{subequations}
 where $\mu = \frac{m_f m_g}{\Omega_{BcS}}$, and $\beta$ is defined in (\ref{eqn:beta}) in Appendix.\ref{app:UB_C}.
 \end{theorem}
 \textit{Proof:} Found in Appendix.\ref{app:UB_C}.
\begin{theorem}
The asymptotic high normalized power for the composite signal UB outage probability is defined as: 
\begin{equation}
P_{out,C}^{UB,\infty} \approx \left( \mathcal{K} \cdot \mathcal{I} \right) \times \left( \frac{\gamma_{th}}{\bar{\gamma}_D} \right)^{m_D} \left( \frac{\gamma_{th}}{\bar{\gamma}_{BcS}} \right)^{m_{min}},
\end{equation}
where $\mathcal{K} = A \cdot B$ is the product of the asymptotic pdfs coefficients found in (\ref{eqn:A}) and (\ref{eqn:B}), $\mathcal{I}$ is a geometric integral, which represents the outage region weighted by the fading exponent as defined in (\ref{eqn:integral_I}) below at the bottom, and $\mathbb{I}(\cdot)$ is the indicator function, which counts whether a specific combination of $u, v, \phi$ is counted or not. While $u$ is the amplitude of the direct link, $v$ is the amplitude of the BcS link, $\phi$ is the phase offset between the two links, and $(|u + v e^{j\phi}|^2 < 1)$ is the outage condition normalized.
\end{theorem}
\begin{figure*}[b]
\hrulefill  
\begin{equation}
\mathcal{I} = \frac{1}{2\pi} \int_0^{2\pi} \int_0^\infty \int_0^\infty \mathbb{I}\left(|u + v e^{j\phi}|^2 < 1\right) \cdot u^{2m_D-1} v^{2m_{min}-1} \, du \, dv \, d\phi.
\label{eqn:integral_I}
\end{equation}
\end{figure*}
\textit{Proof:}
At high normalized power, the outage dominates when there are deep fades, i.e. ($u \to 0, v \to 0$) in (\ref{eqn:Pout_UB}), by using Taylor series, $e^{-\zeta} \approx 1$ as $\zeta \to 0$,  $f_{|h_D|}(u)$ in (\ref{eqn:f_hD}) becomes: 
 \begin{equation}
 f_{|h_D|}(u) \approx \frac{2 m_D^{m_D}}{\Gamma(m_D)} x^{2m_D-1} \approx A u^{2m_D-1},
 \label{eqn:A}
\end{equation}  
where $A = \frac{2 m_D^{m_D}}{\Gamma(m_D)}$. Assuming $m_f \neq m_g$ and $\nu = |m_f - m_g| > 0$, the modified Bessel function $K_\nu(\zeta)$ as $\zeta \to 0$ approximates into \cite{Abramowitz2006}:
\begin{equation}
K_\nu(\zeta) \approx \frac{1}{2}\Gamma(\nu)\left(\frac{\zeta}{2}\right)^{-\nu}. 
\end{equation}
Substituting into (\ref{eqn:f_hBcS}) for $\zeta = 2y\sqrt{\mu}$, $f_{|h_{BcS}|}(v)$ becomes:\\
\begin{equation}
\begin{split}
f_{|h_{BcS}|}(v) \approx \frac{4 \mu^{\frac{m_f+m_g}{2}}}{\Gamma(m_f)\Gamma(m_g)} y^{m_f+m_g-1} \\ \times \left[ \frac{1}{2}\Gamma(|m_f-m_g|) (y\sqrt{\mu})^{-|m_f-m_g|} \right],
\end{split}
\end{equation}
rearranging the exponents of $y$, to get: 
\begin{equation}
f_{|h_{BcS}|}(v) \approx B v^{2m_{min}-1}, \text{where} \, m_{min}= min(m_f,m_g)
\end{equation}  
 \begin{equation}
B=\frac{2 \mu^{m_{min}} \Gamma(|m_f-m_g|)}{\Gamma(m_f)\Gamma(m_g)},
\label{eqn:B} 
 \end{equation}
then extract the normalized power terms of the direct and BcS links from the integral in (\ref{eqn:Pout_UB}), using variable transformation: 
 \begin{equation}
 u = \bar{u} \sqrt{\frac{\gamma_{th}}{\bar{\gamma}_D}}, \quad v =\bar{v} \sqrt{\frac{\gamma_{th}}{\bar{\gamma}_{BcS}}},
\end{equation}  
and substituting these into the expression for $\beta$ found in (\ref{eqn:beta}),\\
\begin{equation}
\beta = \frac{\gamma_{th}(1 - \bar{u}^2 - \bar{v}^2)}{\gamma_{th} 2\bar{u}\bar{v}} = \frac{1 - \bar{u}^2 - \bar{v}^2}{2\bar{u}\bar{v}}.
\end{equation}
Now, by substituting in (\ref{eqn:Pout_UB}) and for simplicity replacing $\bar{u}$ and $\bar{v}$ with $u$ and $v$, respectively, to get (\ref{eqn:Pout_C_inf}) below. 

\begin{figure*}[b] 
\begin{equation}
P_{out,C}^{UB,\infty} = \int_0^\infty \int_0^\infty \Phi(u,v) \left[ A \left(u \sqrt{\frac{\gamma_{th}}{\bar{\gamma}_D}}\right)^{2m_D-1} \right] \left[ B \left(v \sqrt{\frac{\gamma_{th}}{\bar{\gamma}_{BcS}}}\right)^{2m_{min}-1} \right] \sqrt{\frac{\gamma_{th}}{\bar{\gamma}_D}} du \sqrt{\frac{\gamma_{th}}{\bar{\gamma}_{BcS}}} dv,
\label{eqn:Pout_C_inf}
\end{equation}
\end{figure*}
where $\Phi(u,v) = 1 - \frac{\text{Re}\{\arccos(\frac{1-u^2-v^2}{2uv})\}}{\pi}$, grouping $\gamma_{th}$, $\bar{\gamma}_D$ and $\bar{\gamma}_{BcS}$ to get (\ref{eqn:Pout_C_inf2}), found below at the bottom.\\
 
\begin{figure*}[b]
\hrulefill 
\begin{equation}
P_{out,C}^{UB,\infty} \approx \mathcal{K} \cdot \left( \frac{\gamma_{th}}{\bar{\gamma}_D} \right)^{m_D} \left( \frac{\gamma{th}}{\bar{\gamma}_{BcS}} \right)^{m_{min}} \int_0^\infty \int_0^\infty \Phi(u,v)\, u^{2m_D-1} \, v^{2m_{min}-1}\,  du \, dv.
\label{eqn:Pout_C_inf2}
\end{equation}
\end{figure*}


 At high normalized power, the outage is due to either deep fades of $f_{|h_D|} \& f_{|h_{BcS}|}$  or destructive interference of $\Phi(u,v)$. For further refinement, the precise integration bounds in the u,v plane need to capture the outage region. Thus, the geometric integral $\mathcal{I}$ in (\ref{eqn:integral_I}) was introduced, this ends the proof.

Contrary to UB, the ultimate LB should consider the case scenario of constructive interference. This scenario occurs when the phases of the direct link and the backscatter link are perfectly aligned at the CH\footnote{To align the direct and BcS phases constructively at the receiver,  recent research investigated using RIS on the BcS tags itself \cite{Mu2025}, which is out of the scope of this paper.}.

\begin{theorem}
Global LB of the composite signal normalized received power outage probability is presented in (\ref{eqn:C_LB}) below.
\end{theorem}

\begin{figure*}
\begin{equation}
P_{out,C}^{LB} = \left[ \int_0^{\frac{\sqrt{\gamma_{th}}}{\sqrt{\bar{\gamma}_D}}} \left[ \frac{2 m_D^{m_D} u^{2m_D-1} e^{-\frac{m_D u^2}{\Omega_D}}}{\Gamma(m_D)\Omega_D^{m_D}} \right] \\ \left( \int_0^{\frac{\sqrt{\gamma_{th}} - \sqrt{\bar{\gamma}_D}u}{\sqrt{\bar{\gamma}_{BcS}}}} \frac{4 \mu^{\frac{m_f+m_g}{2}} v^{m_f+m_g-1} K_{m_f-m_g}(2v\sqrt{\mu})}{\Gamma(m_f)\Gamma(m_g)} , dv \right) du \right]^N.
\label{eqn:C_LB}
\end{equation}
\end{figure*}

\textit{Proof:} For global LB, we assume perfect phase alignment ($\phi = 0$), by applying to (\ref{eqn:inst_gamma}): 
\begin{equation}
P_{out,C}^{LB} = \mathbb{P}\big( (\sqrt{\bar{\gamma}_D}u_k+ \sqrt{\bar{\gamma}_{BcS}}v_k)  <  \sqrt{\gamma_{th}}\big).
\end{equation}
Rearranging the inequality for \(v_k\) and \(u_k\) yields:
\begin{equation}
v_k < \frac{\sqrt{\gamma_{th}} - \sqrt{\bar{\gamma}_D}u_k}{\sqrt{\bar{\gamma}_{BcS}}}, \, u_k < \frac{\sqrt{\gamma_{th}}}{\sqrt{\bar{\gamma}_D}}.
\end{equation}
Integrating the joint pdf $f_{U,V}(u,v) = f_U(u)f_V(v)$ over the triangular region defined by \(u_k\) and \(v_k\) limits : 
\begin{equation} 
P_{out,C}^{LB,single} = \int_0^{\frac{\sqrt{\gamma_{th}}}{\sqrt{\bar{\gamma}_D}}} \int_0^{\frac{\sqrt{\gamma_{th}} - \sqrt{\bar{\gamma}_D}u}{\sqrt{\bar{\gamma}_{BcS}}}} f_U(u) f_V(v) \, dv \, du,
\label{eqn:P_LB_single} 
\end{equation}
and raising (\ref{eqn:P_LB_single}) to the power of N,        along substituting the explicit PDFs provided in (\ref{eqn:f_hD}) and (\ref{eqn:f_hBcS}) ends the proof.

\subsection{BcS signal outage probability}
\label{subsec:BcS_out}
Under the coherent/SIC detector, the cascaded backscatter SNR is given by:
\begin{equation}
\gamma_{BcS} = \bar{\gamma}_{BcS} |h_{f,k^*}|^2 |h_g|^2,
\end{equation}
where $\bar{\gamma}_{BcS}$ is the BcS link average SNR defined in (\ref{eqn:avg_BcS}), and $k^*$ is the index selected by the symbiotic FAS strategy. Let $W_{sel} = |h_{f,k^*}|$ be the forward link amplitude of the selected port and $G = |h_g|$ be the independent BcS uplink amplitude. Hence, the outage probability of the BcS signal SNR is: 

\begin{equation}
 P_{out,BcS}(\gamma^{D_{d}}_{th}) = \mathbb{P}( G^2 < \frac{1}{W_{sel}^2}{\frac{\gamma^{D_{d}}_{th}}{\bar{\gamma}_{BcS}}}),
\end{equation}
where $\gamma^{D_{d}}_{th}$ is the minimum permissible SNR threshold for BcS data detection\footnote{For concise equation presentation, we omit the threshold prefix $(D_{d})$ in the upcoming equations.}. Since $G$ is independent of the selection process, we condition on $W_{sel}$:
\begin{equation}
P_{out,BcS} = \int_{0}^{\infty} F_{G^2}\left(\frac{1}{w^2} {\frac{\gamma_{th}}{\bar{\gamma}_{BcS}}}|W_{sel}=w\right) f_{W_{sel}}(w) \, dw,
\label{eqn:P_out_BcS}
\end{equation}
finding the exact pdf of $W_{sel}$, i.e. $f_{W_{sel}}(w)$ for each of the proposed symbiotic FAS selection strategies is a tedious process specially with  correlated ports channels. Thus, the global UB and LB equations is computed next.\\

\begin{theorem}
Global UB of BcS signal SNR outage probability is defined in (\ref{eqn:Pout_BcS_UB}). 
\end{theorem}

\begin{figure*}[t]
\begin{equation}
P_{out,BcS}^{UB}(\gamma_{th}) = \int_{0}^{\infty} \underbrace{\frac{\gamma\left(m_g, \frac{m_g}{\Omega_g}  \frac{\gamma_{th}}{w^2 \bar{\gamma}_{BcS}}\right)}{\Gamma(m_g)}}_{\text{uplink Outage given } w} \cdot \underbrace{\left[ \frac{2 m_f^{m_f}}{\Gamma(m_f) \Omega_f^{m_f}} w^{2m_f-1} e^{-\frac{m_f}{\Omega_f} w^2} \right]}_{\text{forward link pdf}} \, dw.
\label{eqn:Pout_BcS_UB}
\end{equation}
\end{figure*}

\textit{Proof:}
UB exists under full correlation among FAS ports. Thus, the pdf of the forward link amplitude $w=|h_f|$ simplifies to Nakagami-m distribution: 
 
\begin{equation}
f_{|h_f|}(w) = \frac{2 m_f^{m_f}}{\Gamma(m_f) \Omega_f^{m_f}} w^{2m_f-1} e^{( -\frac{m_f}{\Omega_f} w^2 )},
\label{eqn:f_hf}
\end{equation}
and the cdf of the uplink power $G^2$ is Gamma distribution:
\begin{equation}
F_{G^2}(g) = \frac{\gamma\left(m_g, \frac{m_g}{\Omega_g} g\right)}{\Gamma(m_g)}.
\label{eqn:cdfG}
\end{equation} 
Substituting (\ref{eqn:cdfG}) and (\ref{eqn:f_hf}) into (\ref{eqn:P_out_BcS}) ends the proof. \\

\begin{theorem}
Global LB of BcS signal SNR outage probability is defined in (\ref{eqn:Pout_BcS_LB}).
\end{theorem}

\begin{figure*}[t]
\begin{equation}
P_{out,BcS}^{LB} (\gamma_{th}) = \int_{0}^{\infty} \frac{\gamma\left(m_g, \frac{m_g}{\Omega_g}  \frac{\gamma_{th}}{w^2 \bar{\gamma}_{BcS}}\right)}{\Gamma(m_g)} \cdot \left( N \left[ \frac{\gamma(m_f, \frac{m_f}{\Omega_f} w^2)}{\Gamma(m_f)} \right]^{N-1} \cdot \frac{2 m_f^{m_f}}{\Gamma(m_f) \Omega_f^{m_f}} w^{2m_f-1} e^{( -\frac{m_f}{\Omega_f} w^2 )}\right) dw.
\label{eqn:Pout_BcS_LB}
\end{equation} 
\end{figure*}

\textit{Proof:}
LB associates with the maximization of the forward link amplitude $w = |h_{f,k}|$ across $\{k \in N\}$ ports,
given the cdf of a single Nakagami-m amplitude $|h_f|$ is,  
\begin{equation}
F_{|h_f|}(w) = \frac{\gamma(m_f, \frac{m_f}{\Omega_f} w^2)}{\Gamma(m_f)}.
\label{eqn:cdf_hf}
\end{equation}
Hence, for N ports the pdf of the maximum forward link is
\begin{equation}
f_{w_{max}}^{LB}(w) = \frac{d}{dw} [F_{|h_f|}(w)]^N = N \cdot [F_{|h_f|}(w)]^{N-1} \cdot f_{|h_f|}(w),
\label{eqn:fymax}
\end{equation}
by substituting (\ref{eqn:cdf_hf}) into (\ref{eqn:fymax}), the pdf of maximum forward link is estimated, then substituting along with (\ref{eqn:cdfG}) in (\ref{eqn:P_out_BcS}) ends the proof.

\begin{theorem}
UB $\&$ LB asymptotic high SNR of BcS signal outage probability in the integral form are presented in (\ref{eqn:outage_UB_asym}) and (\ref{eqn:outage_LB_asym}). 
\end{theorem}

\begin{figure*}
\begin{equation}
P_{out, BcS}^{UB, \infty} \approx \left( \mathcal{C} \int_{0}^{\infty} w^{-2m_g}\left[ \frac{2 m_f^{m_f}}{\Gamma(m_f) \Omega_f^{m_f}} w^{2m_f-1} e^{-\frac{m_f}{\Omega_f} w^2} \right] \, dw \right) \cdot \gamma_{th}^{m_g},
\label{eqn:outage_UB_asym}
\end{equation}
\begin{equation}
P_{out, BcS}^{LB, \infty} \approx \left( \mathcal{C} \int_{0}^{\infty} w^{-2m_g} \left( N \left[ \frac{\gamma(m_f, \frac{m_f}{\Omega_f} w^2)}{\Gamma(m_f)} \right]^{N-1} \cdot \frac{2 m_f^{m_f}}{\Gamma(m_f) \Omega_f^{m_f}} w^{2m_f-1} \exp\left( -\frac{m_f}{\Omega_f} w^2 \right) \right) \, dw \right) \cdot \gamma_{th}^{m_g}.
\label{eqn:outage_LB_asym}
\end{equation}
\end{figure*}

\textit{Proof:} In the high SNR domain, as $\gamma_{th} \to 0$, the asymptotic behaviour is dominated by the cdf near zero. Hence, derived using the BcS outage UB and LB in (\ref{eqn:Pout_BcS_UB}) and (\ref{eqn:Pout_BcS_LB}), respectively.
 
Given the lower incomplete gamma function $\gamma(s, x) = \int_0^x t^{s-1} e^{-t} dt$ can be approximated for small arguments ($x \to 0$) as \cite{Jameson2016}:
\begin{equation}
\gamma(s, x) \approx \frac{x^s}{s},
\end{equation} 
consequently, the cdf term for the uplink outage $F_{G^2}$ in (\ref{eqn:P_out_BcS}), which also appears in (\ref{eqn:Pout_BcS_UB}) and (\ref{eqn:Pout_BcS_LB}), approximates to (\ref{eqn:gamma_high}).\\

\begin{figure*}[t]
\begin{equation}
\frac{\gamma\left(m_g, \frac{m_g \gamma_{th}}{\Omega_g w^2 \bar{\gamma}_{BcS}}\right)}{\Gamma(m_g)} \approx \frac{1}{\Gamma(m_g) \cdot m_g} \left( \frac{m_g \gamma_{th}}{\Omega_g w^2 \bar{\gamma}_{BcS}} \right)^{m_g} = \frac{1}{\Gamma(m_g+1)} \left( \frac{m_g}{\Omega_g \bar{\gamma}_{BcS}} \right)^{m_g} \gamma_{th}^{m_g} w^{-2m_g}.
\label{eqn:gamma_high}
\end{equation}
\hrulefill 
\end{figure*}

Let the constant term be $\mathcal{C} = \frac{1}{\Gamma(m_g+1)} \left( \frac{m_g}{\Omega_g \bar{\gamma}_{BcS}} \right)^{m_g}$. By substituting in (\ref{eqn:P_out_BcS}) and rearranging:
\begin{equation}
P_{out, BcS}^{UB, \infty} \approx \left( \mathcal{C} \int_{0}^{\infty} w^{-2m_g} f_{|h_f|}(w) \, dw \right) \cdot \gamma_{th}^{m_g}.
\end{equation}
Similarly, by substituting in the BcS outage LB  (\ref{eqn:Pout_BcS_LB}), 
\begin{equation}
P_{out, BcS}^{LB, \infty} \approx \left( \mathcal{C} \int_{0}^{\infty} w^{-2m_g} f_{W_{max}}(w) \, dw \right) \cdot \gamma_{th}^{m_g},
\end{equation}
substituting with $f_{|h_f|}(w)$ and $f_{W_{max}}(w)$ ends the proof.

\textbf{Remark:} \textit{The Nakagami ($m_g$) coefficient of the single uplink A-IoT channel $(h_g)$ to the CH represents the diversity gain of both UB and LB equations, as it's the power of threshold ($\gamma_{th}^{m_g}$) in both equations. Hence, it consists the bottleneck of the system no matter what was the number of FAS ports ($N$). At the same time, maximizing the BcS link with FAS provided an asymptotic coding gain of over the single fixed antenna (i.e. UB) gain, which guarantees a smaller outage probability value.}

\begin{corollary}
 The theoretical possible BcS asymptotic coding gain (in dB) provided by the symbiotic FAS Selection strategies, approximated by:`
\begin{equation}
G_{code} = 10 \log_{10} \left( \frac{P_{out, UB}^{\infty}}{P_{out, LB}^{\infty}} \right) = 10 \log_{10} \left( \frac{\mathcal{M}_{UB}}{\mathcal{M}_{LB}} \right),
\end{equation}
where $\mathcal{M}_{LB} = \int_{0}^{\infty} w^{-2m_g} f_{W_{max}}(w) dw$ is the symbiotic FAS gain and $\mathcal{M}_{UB} = \int_{0}^{\infty} w^{-2m_g} f_{|h_f|}(w) dw$ is the single antenna gain.
\end{corollary}

\subsection{Coexistence outage probability (COP)}
By isolating the exact mathematical global UB and LB for both $P_{out, C}$ and $P_{out, BcS}$ in the preceding subsections, we have inherently mapped the absolute geometric envelope bounds of the system's coexistence capability.
\begin{theorem}
\label{theo:coexistence}
The joint COP is tightly bounded by the individual composite and BcS performance metrics as follows:
\begin{itemize}
\item{LB governed by the dominant worst performing link:
\begin{equation}
P_{out, coex}^{Z_k^*} \ge \max\left(P^{Z_k^*}_{out, C}, P^{Z_k^*}_{out, BcS}\right).   
\end{equation}}
\item{UB represents the worst case mutually exclusive scenario:
\begin{equation}
 P_{out, coex}^{Z_k^*} \le \min\left(1, P^{Z_k^*}_{out, C} + P^{Z_k^*}_{out, BcS}\right),
\end{equation}}
\end{itemize}
where ${Z_k^*}$ represents the symbiotic FAS selection criteria.
\end{theorem} 

\textit{Proof:} Find in Appendix \ref{app:cop}.\\

This decouples the fundamental limits of the UAV physical channel from the specific operational weights assigned to data or energy as needed by the IoT network.

\section{ Joint optimization of COP}

In this section, the reliability of the symbiotic assisted FAS-UAV network is formulated as a global joint optimization framework in (\ref{eqn:global_opt}). The joint optimization depends on two key factors; the UAV optimal placement $\mathbf{p}_{\text{UAV}}$ which effects the large scale path loss fading and FAS port selection matrix $\mathcal{K}^*$ that balances the small-scale fast fading effect through tackling the trade-off multi-objective optimization of composite signal outage $\mathbb{P}_{out, C}$ ($J_1$) and BcS signal outage $\mathbb{P}_{out, BcS}$ ($J_2$) within the symbiotic network. Thus, this coupled global framework can be decoupled without loss of optimality into two sequential sub-problems; first the microscopic realization-level ($M$) port matrix ($\mathcal{K}^*$) optimization where we derive the absolute non-convex Pareto-front boundary for any arbitrary UAV spatial coordinate, then the macroscopic spatial UAV placement $\mathbf{p}_{\text{UAV}}$ optimization where we determine the optimal 2D coordinates $(x_U^*, y_U^*)$ for a fixed UAV altitude $H$ that minimize the expected multi-objective COP profile across the network's geographic layout.

\begin{subequations}
\label{eqn:global_opt}
\begin{align}
\min_{\mathbf{p}_{\text{UAV}}, \mathcal{K}^*} \quad \left\{ J_1(\mathbf{p}_{\text{UAV}}, \mathcal{K^*}), \; J_2(\mathbf{p}_{\text{UAV}}, \mathcal{K^*}) \right\}\\ \text{s.t.} \quad x_{min} \le x_U \le x_{max}, \\
\quad y_{min} \le y_U \le y_{max}, \\k_m \in \mathcal{N}, \quad \forall m \in \{1, 2, \dots, M\},
\end{align}
\end{subequations}





\subsection{Pareto front multi-objective optimization}
For a fixed large-scale channel  state, standard symbiotic multi-objective frameworks frequently deploy the linear weighted sum method \cite{Li2026FASR} \cite{Wang2022}. Despite its computational simplicity, linear scalarization is incapable of tracing the concave portions of FAS-SR's non-convex Pareto front boundary. Therefore, we deploy the weighted Tchebycheff scalarization framework: \begin{subequations}
\begin{align}
&\min_{\mathcal{K^*}} \max \left\{ \lambda \left( J_1(\mathcal{K^*}) - J_1^* \right), \; (1 - \lambda) \left( J_2(\mathcal{K^*}) - J_2^* \right) \right\}\\&\text{s.t.} \quad k_m \in \mathcal{N}, \quad \forall m \in \{1, 2, \dots, M\},
\end{align}
\end{subequations}
where $J_1^*$ and $J_2^*$ represent the respective ideal utopian outage values. However, to establish a definitive baseline that directly reflects the maximum possible reliability region, the $\epsilon$-constraint method is introduced: 
\begin{subequations}
\begin{align}
&\min_{\mathcal{K^*}} \quad J_1(\mathcal{K^*})\\
&\text{s.t.} \quad J_2(\mathcal{K^*}) \le \epsilon,\\ 
&k_m \in \mathcal{N}, \quad \forall m \in \{1, 2, \dots, M\}, 
\end{align}
\end{subequations}
where $\epsilon \in [0, 1]$ represents the maximum tolerable backscatter outage probability. By sweeping $\epsilon$ over $G_{\epsilon}$ distinct threshold values, this formulation yields a highly accurate Pareto-front curve for the reliability region. Nevertheless, solving this globally coupled NP-hard problem across all $M$ channel realizations simultaneously has a complexity of $ \mathcal{O}(G_{\epsilon} \cdot N^M)$, which creates a computational bottleneck that is unviable for real-time UAV operations. By contrast, our proposed symbiotic FAS selection strategies has a linear $\mathcal{O}(MN)$ complexity as it completely avoids the global coupling and grid search computations.
\subsection{ Optimal symbiotic FAS-UAV placement}
The ultimate UAV placement optimization problem to minimize the joint coexistence outage probability ($P_{out,coex}$) under the symbiotic FAS port matrix $\mathcal{K^*}$ as:
\begin{subequations}
\label{eqn:P1} 
\begin{align}
\quad \min_{\mathbf{p}_{\text{UAV}}} \quad \max \left( \mathbb{P}_{out, C}(\mathbf{p}_{\text{UAV}}), \; \mathbb{P}_{out, BcS}(\mathbf{p}_{\text{UAV}}) | \mathcal{K^*} \right)  \label{eqn:P1_min}   \\ 
 \text{s.t.} \quad x_{min} \le x_U \le x_{max}, \quad y_{min} \le y_U \le y_{max},   \label{eqn:min_xy} 
\end{align}
\end{subequations}
where ${Z_k^*}$ represents the symbiotic FAS selection criteria. Since the objective function is highly non-convex due to the integration of Nakagami-$m$ fading and FAS correlation bounds, a 2D exhaustive grid search is the most robust method to find the global optimum and visualize the spatial effects.

 \section{Simulation Results}

 In Monte Carlo simulations, a UAV is deployed at an altitude of 100 m such that $\mathbf{p}_{\text{UAV}} = [0, 0, 100]^T \text{ m}$, a sensor is positioned near the ground level at $\mathbf{p}_{\text{Tag}} = [20, 20, 1.5]^T \text{ m}$, and the CH acts as a stationary anchor node at $\mathbf{p}_{\text{CH}} = [50, 0, 10]^T \text{ m}$. The wireless channels are modelled with Nakagami fading factor $m=2$ for both the direct link $h_D$ and BcS forward link $h_f$ while $m=1$ for BcS backward link $h_g$, a path loss exponent of $\alpha = 2.5$ and the receiver noise floor is set to $\sigma ^{2}=-90\text{ dBm}$ \cite{Abuzgaia2026FAS}. The received power threshold for EH activation is set to $\gamma^{pr}_{th} = 60 \text{dB}$\footnote{By substitution in $(P_{r}^{min}\text{\ (dBm)}=\gamma _{th}^{pr}\text{\ (dB)}+\sigma ^{2}\text{\ (dBm)}-10\log _{10}(\eta ))$ for an efficiency of $(\eta = 0.8)$, $\gamma^{pr}_{th}$ corresponds to an absolute EH circuit sensitivity of $ P_{r}^{min}\approx -29\text{dBm}$. Which matches modern ultra-low-power RF energy harvesting circuits threshold range; $(-20  \, \text{to} -50)\, \text{dBm}$ \cite{Chen2024,Yin2025}.}, whereas the BcS data decoding SNR threshold is set to $\gamma^{Dd}_{th} = 0\text{ dB}$ and BcS reflection efficiency $\vartheta=0.6$ \cite{Ali2026}. \\

 Fig.\ref{fig:out_theo_corr}(a) demonstrates the accuracy of the derived outage probability bounds and asymptotic equation of the composite signal normalized received power at the CH along the outage curves of the novel symbiotic FAS strategies. The simulated fixed single antenna curve was in identical match with the UB curve as was expected. MCGS strategy showed the best performance although it has a 400 shift in the linear scale from the LB curve due to the correlation effect of the FAS limited length $0.2\lambda$. On the other hand, TAPS demonstrated identical performance with MCGS for lower values of outage till approximately $10^{-3}$ the outage curve bent and JBS demonstrate enhanced value. This is explained by the heuristic optimization nature of TAPS compared with the other symbiotic strategies, as TAPS chooses the activated port based on priority   of fulfilling the BcS link threshold which doesn't always guarantee that the composite link outage would be optimum. While, Fig.\ref{fig:out_theo_corr}(b) proved the derivation exact match of UB and LB theoretical with the asymptotic equations of the BcS signal SNR outage performance. The MCGS curve was identical to the simulated single fixed antenna system and the UB curve, while the TAPS strategy showed the best performance which emphasise the importance of symbiotic FAS techniques that balances the BcS data detection with the composite signal energy harvesting. A total coding gain of almost 6.2dB was demonstrated although not all of this coding gain was harnessed by TAPS strategy due to the high correlation effect of FAS limited length $0.2\lambda$.\\

\begin{figure}[t]
 \centering
  \begin{subfigure}[]
  \centering
 \includegraphics[width=7cm]{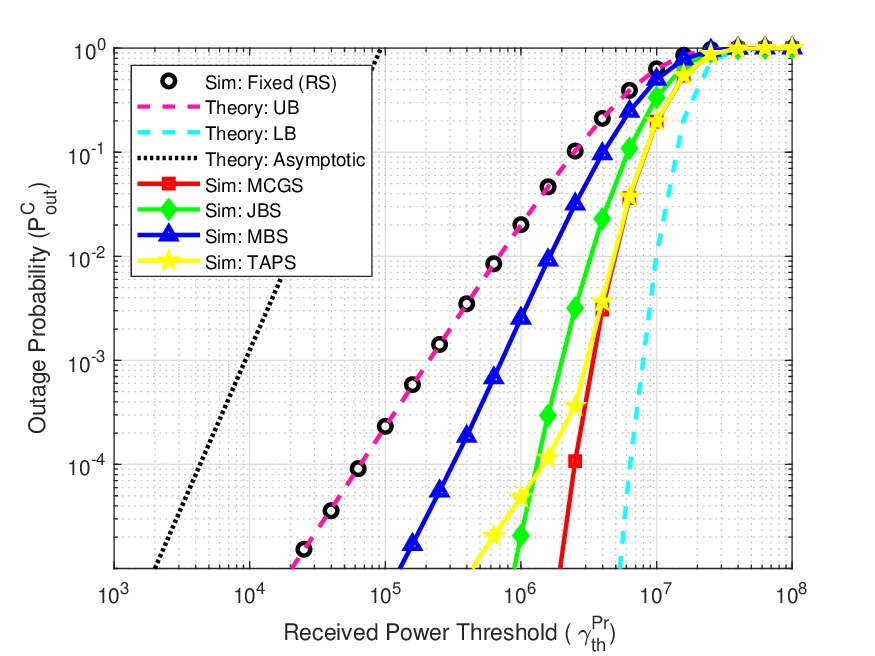}
\label{fig:out_C_theo}
\end{subfigure}
 \hfill
 \begin{subfigure}[]
 \centering
 \includegraphics[width=7cm]{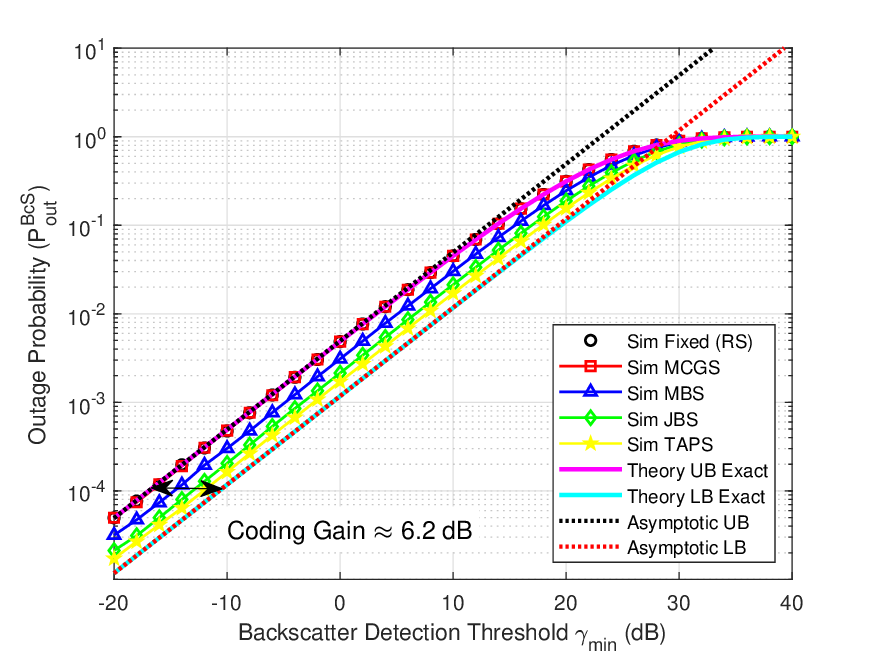}
\label{fig:out_BcS_theo}
\end{subfigure}
 \caption{Outage probability vs. the threshold showing UB, LB, asymptotic bounds, simulation curves of the symbiotic FAS strategies and RS strategy for single fixed antenna system when FAS ports $N =10$, FAS length $W=0.2\lambda$ and the transmit power $P_t=30dBm$, for (a) Composite normalized signal received power and (b) BcS signal SNR.}
 \label{fig:out_theo_corr}
 \end{figure}

\begin{figure}[htbp]
    \centering
    \begin{minipage}{0.5\textwidth}
    \centering
    \begin{subfigure}[]
        \centering
        \includegraphics[width=\textwidth]{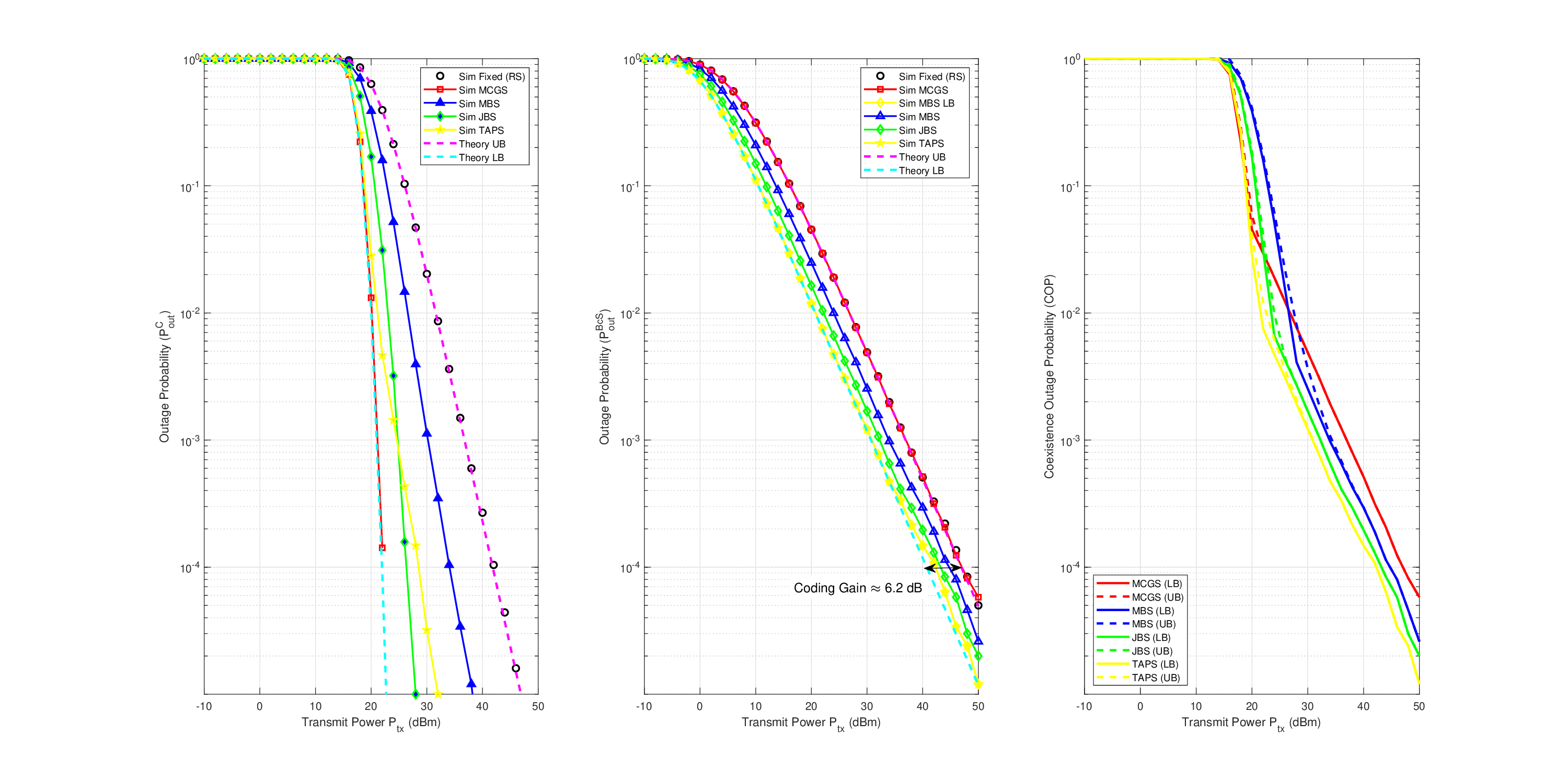}
        \label{fig:sub-first}
    \end{subfigure}
    \hfill
    \begin{subfigure}[]
        \centering
        \includegraphics[width=\textwidth]{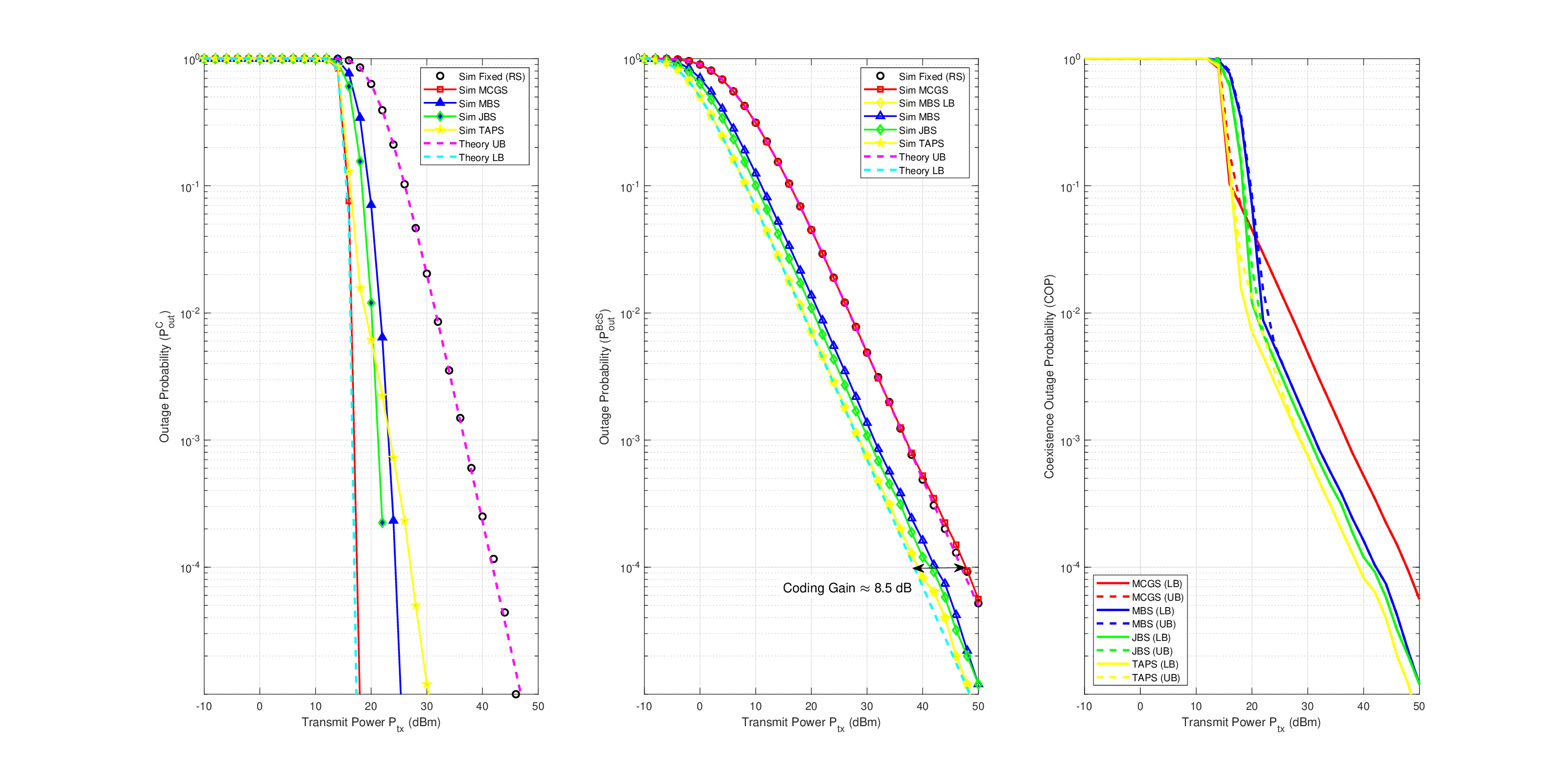}
        \label{fig:sub-second}
    \end{subfigure}
    \end{minipage}
    \caption{Outage probability vs. the transmit power in dBm of the composite signal, BcS link and COP under different symbiotic FAS port strategies for FAS length $W=2\lambda$ and FAS ports $(a)N =10$ and $(b)N=100$.}
\label{fig:out_sub_N_10}
\end{figure}

Fig.\ref{fig:out_sub_N_10} illustrate the effect of increasing the number of FAS ports from 10 to 100 on the outage probability of the composite normalized received power, the BcS signal SNR along with the LB of COP. For the composite outage, MCGS with FAS length of $W=2\lambda$ was the optimal strategy reaching to almost the theoretical LB as was proved in $P_{out, C}^{LB, MCGS}$ theoretical derivation, in contrary to Fig.\ref{fig:out_theo_corr} miniatured length FAS deterioration in performance. Increasing FAS ports has increased all the symbiotic strategies outage curves slops due to the diversity gain. Non the less, TAPS being a heuristic optimization technique rather a direct maximum port selection showed a transition in the slop, as after $10^{-3}$ both of MBS and JBS performed better. On the other hand, the BcS signal outage curves slops were identical for all; the UB, LB and for increasing ports numbers. Equalling to Nakagami m factor of the BcS backward link $m_g$, serving as the bottleneck of the symbiotic system, as was proven by the asymptotic analysis. Nevertheless, a FAS coding gain increment from 6.2dB at N=10 to 8.5dB at N=100 was demonstrated, as the curves at N=100 showed a further shift to the left side. Regarding symbiotic strategies, MCGS showed the worst performance for the BcS signal where it matches with the single antenna outage curve, while obviously both of TAPS and the simulated curve of $\text{MBS}_{LB}$ showed the best performance matching the derived theoretical LB, and JBS strategy proved a better performance than MCGS and MBS as it balances the direct link and BcS link when choosing the activated port. The normalized scale for the composite signal power threshold aided us to  visualize the symbiotic bottleneck where the energy link and data link cross as it appears as a sudden bend in the COP curve, where increasing the number of ports N has moved further to the left due to the coding gain effect on the BcS link and upward due to the diversity gain effect on the composite signal link. Moreover, the TAPS strategy showed the best COP performance which indicate the importance of balanced optimization for both of the involved links. \\

\begin{figure}[t]
 \centering
 \includegraphics[width=7cm]{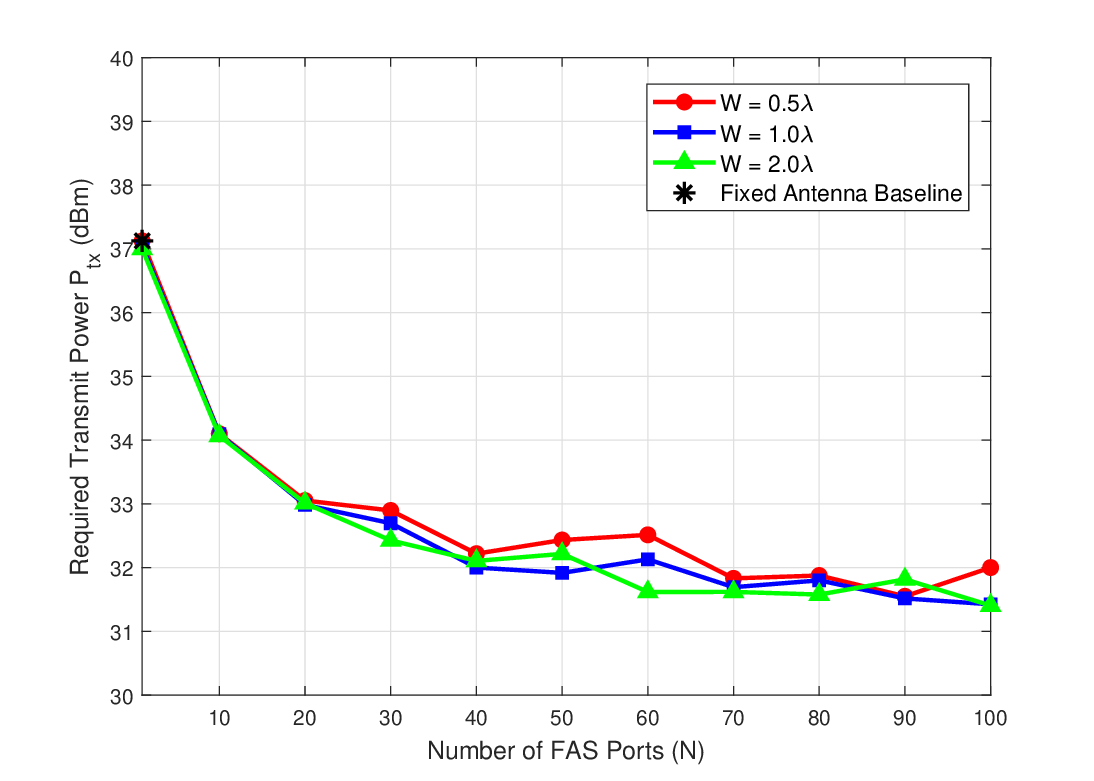}
 \caption{FAS Gain through transmit power $\text{P}_{tx}$ required for a LB bottleneck of $\text{COP} = 10^{-3}$ vs. number of FAS ports $\text{N}=(1,100)$ under MBS strategy for different FAS length $\text{W}=(0.5,1,2)\lambda$ and single fixed antenna system.}
\label{fig:fas_gain}
\end{figure}

Fig.\ref{fig:fas_gain} illustrates the required UAV transmit power $P_{tx}$ to achieve a target coexistence outage probability (COP) of $\epsilon = 10^{-3}$ versus the number of FAS ports $N$. Compared to the legacy fixed-antenna baseline ($N=1$), which demands a high power of $37.1\text{ dBm}$, shifting to FAS triggers a sharp exponential decline in the required $P_{tx}$. At $N=20$, the necessary power drops to approximately $33\text{ dBm}$ across all normalized sizes $W$, yielding a massive RF power savings of $\approx 4\text{ dB}$ that translates to $60\%$ reduction in the  transmitted power. Crucially, this power reduction directly translates to a significant extension of the battery-limited flight time for payload-constrained UAV chassis. Furthermore, the tightly bound convergence of the curves demonstrates that the high spatial resolution of the ports dominates over the raw physical dimensions of the antenna array. Even under limited spatial length of $W = 0.5\lambda$, the FAS structure captures substantial spatial diversity and coding gain, proving their capabilities in supporting SWAP limited UAVs and surpassing the bulky, heavy multi-antenna hardware. Diminishing returns emerge beyond $N=40$, establishing a clear design threshold for lightweight, energy-efficient, and aerodynamically optimized symbiotic UAVs platforms.\\

\begin{figure}[t]
 \centering
 \includegraphics[width=7cm]{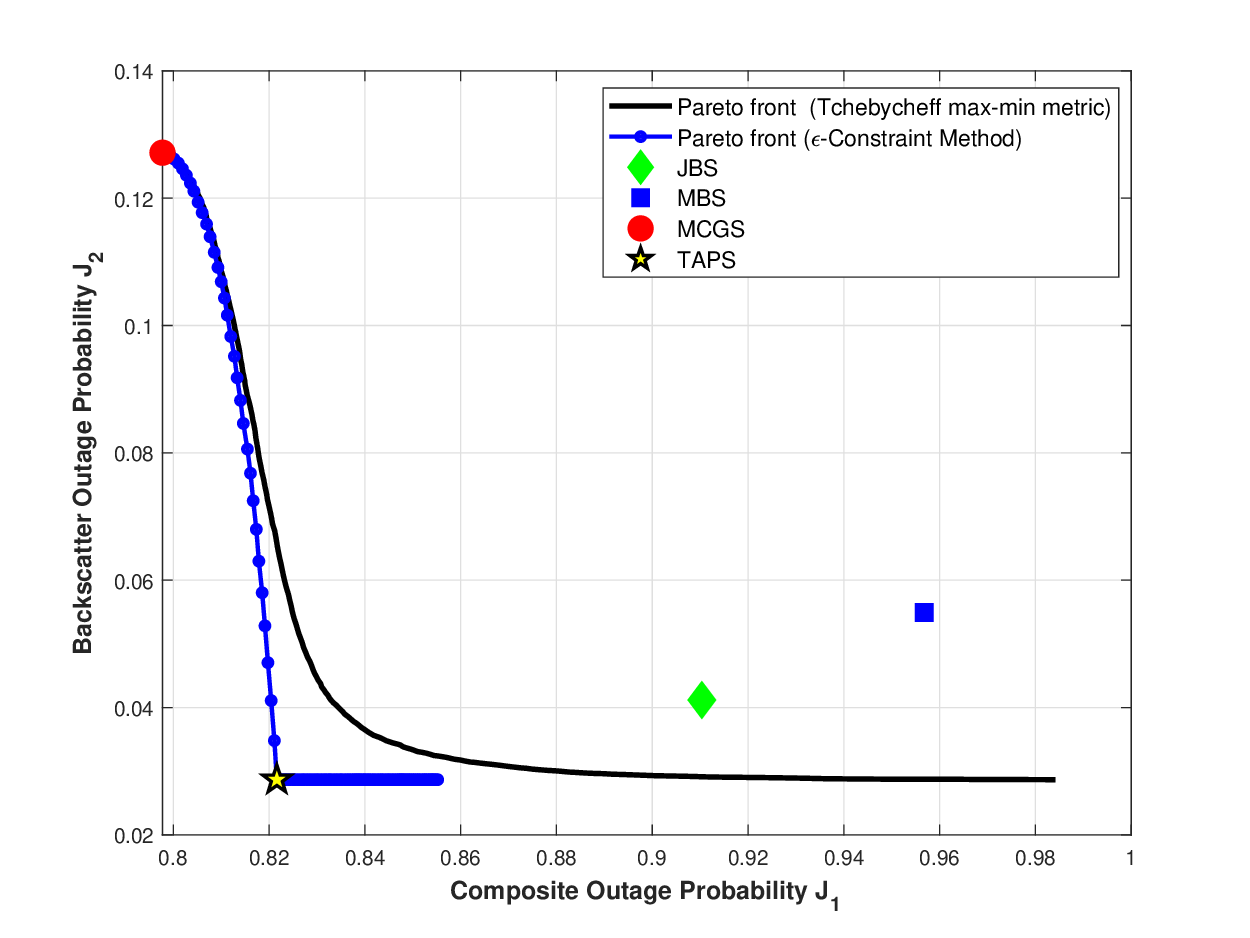}
 \caption{Pareto front boundary with two methods (Tchebycheff max-min and $\epsilon$ constraint) along with the operating points of the novel proposed symbiotic FAS schemes when outage thresholds were ( $\bar{\gamma}^{Pr}_{th}=75dB$, $\bar{\gamma}^{Dd}_{th}=15dB$), FAS length and ports were $W=1\lambda$ and $N=30$, respectively.}
\label{fig:Pareto_front}
\end{figure}

Fig.\ref{fig:Pareto_front} highlights the operating points of the novel symbiotic FAS strategies, specifically the tightness of our low-complexity single step TAPS scheme against the $\epsilon$-constraint frontier boundary. The tightness of TAPS toward the boundary knee demonstrates optimal symbiotic reliability performance compromise  without requiring the computation complexity of multi-objective solvers. Moreover, the Pareto boundary highlights the inescapable physical trade-off between the primary EH link and the secondary BcS link, demonstrated through the position of MCGS on the rightmost anchor of the front prioritizing the composite signal reliability compared with MBS position at the center-left prioritizing the BcS reliability, while JBS showed a better performance compromise although not ideal as TAPS. \\

Fig.\ref{fig:FAS_feasability} illustrates the maximum permissible operational UAV velocity ($v_{\text{max}}$) as a function of FAS port dimension ($N$) under a control overhead budget of $\eta_{\text{max}} = 10\%$. As analytically derived in our asymptotic complexity analysis, the maximum velocity curves scale according to $\mathcal{O}(1/N)$ for both evaluated carrier frequencies. Nevertheless, when operating at a lower carrier frequency ($f_c = 2.4\text{ GHz}$), the relaxed Doppler constraints allow the UAV to maintain significant physical mobility (up to $v_{\text{max}} \approx 40.5\text{ m/s}$ at $N = 30$ ports). Conversely, when moving to the $5.8\text{ GHz}$ band, the heightened Doppler sensitivity contracts the available coherence window, compressing the mobility to less than the half. This visualization serves as a crucial design guideline, proving that the low-complexity, non-iterative nature of the proposed symbiotic FAS schemes guarantee physical real-time feasibility under realistic drone speeds and dense port deployments.\\

Fig.\ref{fig:2d_uav_placement} presents the spatial heat map with a contour plot of the LB COP under the MBS strategy across the operational grid of ($x=-10:80$) and ($y=-20:50$). The optimal UAV position of $(X_U = 14.0, Y_U = 16.0)$ with the minimum achievable COP of $P_{out,coex}=1.55e^{-03}$ demonstrates that the global minimum outage is aggressively skewed toward the IoT tag, hence the distance dependent path loss of the BcS forward link $\text{d}_{\text{f}}$ is the primary bottleneck. By optimizing the UAVs placement, we autonomously overcome this tag deficit ensuring the total symbiotic network performance.

\begin{figure}[t]
\centering
 \includegraphics[width=7cm]{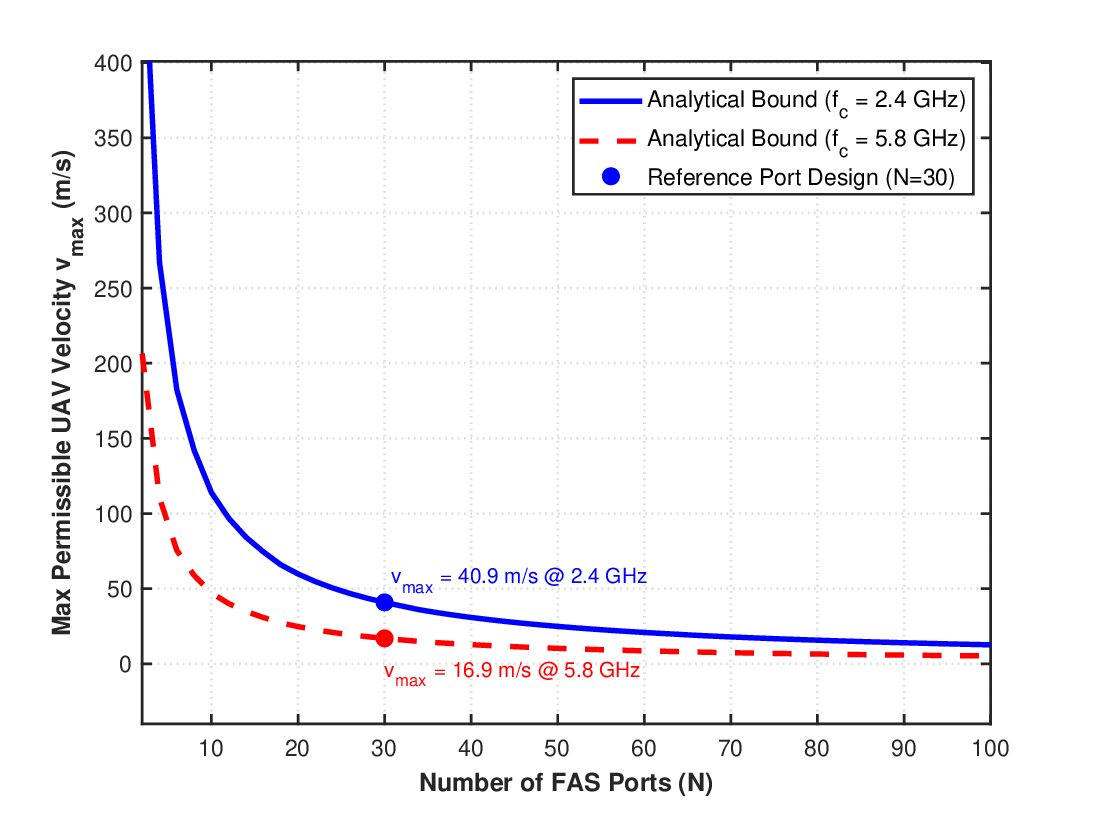}\vspace{-0.2cm}
 \caption{Mobility operational envelopes with maximum protocol overhead fraction of ($\eta_{max} = 10\%$) vs. number of FAS ports and 
 under different operating frequencies. }
\label{fig:FAS_feasability}
\end{figure}

 \begin{figure}[t]
\centering
 \includegraphics[width=7cm]{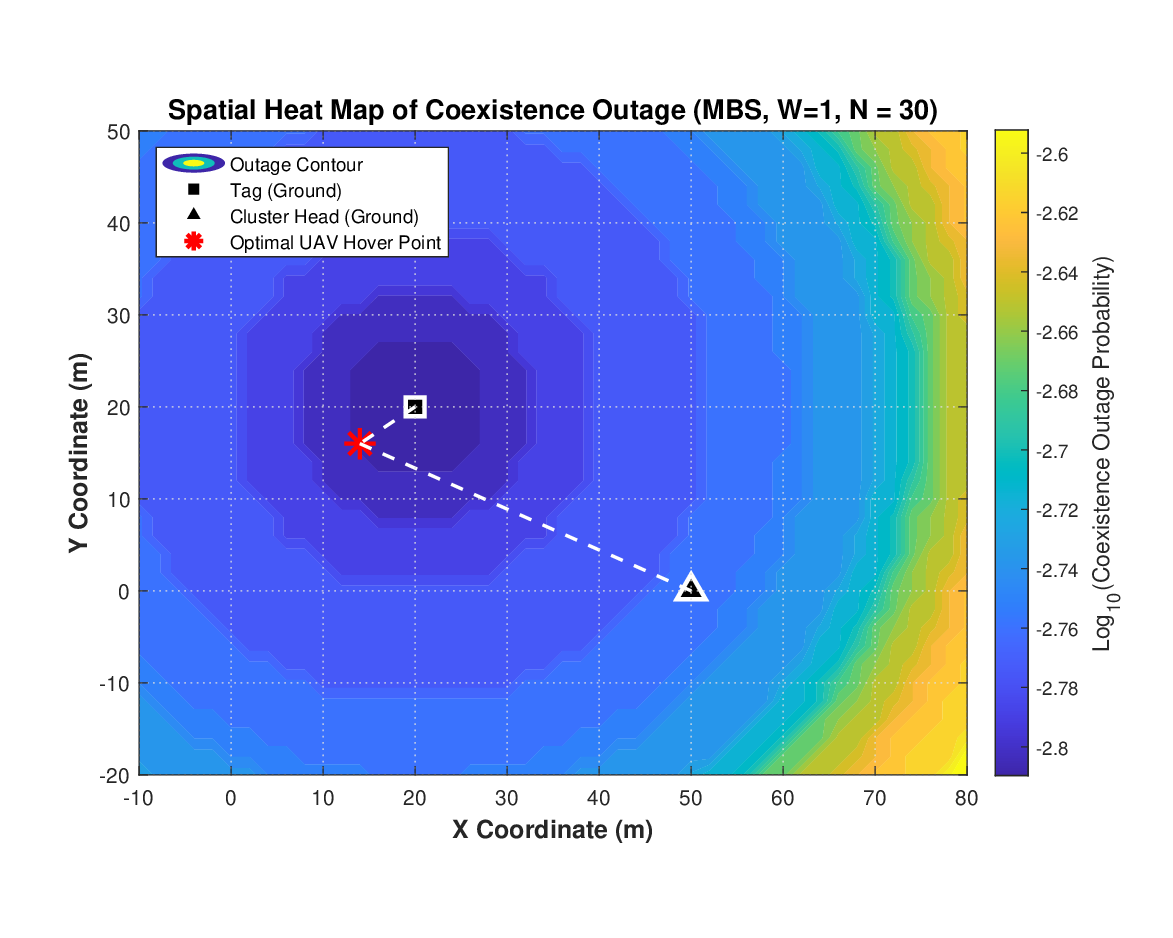}\vspace{-0.2cm}
 \caption{Spatial heat map of COP showing the optimal horizontal UAV placement with fixed altitude of $100m$ for MBS strategy, when FAS length $ W=1$ and ports $N = 30$.}
\label{fig:2d_uav_placement}
\end{figure}

\section{Conclusion}
This paper introduced novel FAS port selection strategies for UAVs assisted symbiotic 6G networks, focusing on optimizing the energy harvesting and BcS signals reliability by achieving the optimal Pareto boundary performance without the classical multi-objective solver exponential complexity. The derived analytical and asymptotic bounds, and the Monte-Carlo simulations proved that these strategies significantly improve the outage probability, diversity and coding gain. The feasibility under different operating frequencies and UAVs velocities was demonstrated and the optimal UAVs placement surpassed the bottleneck link and further enhanced the system reliability. Posing symbiotic FAS mounted UAVs as a viable energy efficient scheme for green 6G networks. 
\appendices
\section{UB of the composite signal}
\label{app:UB_C}
Composite signal phase offset doesn't depend on the selected port, then conditioning on fixed amplitudes $u$ and $v$:
\begin{equation}
\begin{split}
F_{Z|u,v}(\gamma_{th}) = \\ \text{P}\left( \bar{\gamma}_D u^2 +\bar{\gamma}_{BcS} v^2 + 2\sqrt{\bar{\gamma}_D \bar{\gamma}_{BcS}} uv \cos(\Delta\phi) < \gamma_{th} \right).
\end{split}
\end{equation}
Solving for $\cos(\Delta\phi)$:
\begin{equation}
\cos(\Delta\phi) < \beta, \quad \text{where } \beta = \frac{\gamma_{th} -\bar{\gamma}_D u^2 -\bar{\gamma}_{BcS} v^2}{2\sqrt{\bar{\gamma}_D \bar{\gamma}_{BcS}} uv},
\label{eqn:beta}
\end{equation}
and since $\Delta\phi$ is uniform, the probability is determined by the angle:\\
\begin{equation}
F_{Z|u,v}(\gamma_{th}) = \begin{cases}
1 & \text{if } \beta \ge 1  \\
0 & \text{if } \beta \le -1  \\
1 - \frac{\arccos{(\beta)}}{\pi} & \text{if} -1 < \beta < 1.
\end{cases}
\label{eqn:F_Z}
\end{equation}
Integrating (\ref{eqn:F_Z}) over $f_{UV}= f_U f_V$, which is the joint pdf of the independent direct and BcS channels given in (\ref{eqn:f_hD}) and (\ref{eqn:f_hBcS}), respectively, ends the proof.
\section{COP}
\label{app:cop}
Operational coexistence requires the simultaneous satisfaction of the EH and BcS data retrieval constraints, as: 
\begin{equation}
P_{out, coex} = \text{Pr}\Big( S_{C,k^*} < \gamma^{Pr}_{th} \ \cup \ S_{BcS,k^*} < \gamma^{Dd}_{th} \Big), 
\label{eqn:union}
\end{equation}
where $k^*$ represents the activated port under a given FAS selection strategy.
Thus, invoking the Fréchet-Hoeffding probability inequalities to the union in (\ref{eqn:union}) \cite{Hoppe2008}, ends the proof.

\bibliographystyle{IEEEtran}
\bibliography{FAS_UAV_SR} 
\end{document}